\documentclass[accepted]{uai2025} 
                        
\usepackage[american]{babel}

\usepackage{natbib} 
\usepackage{mathtools} 
\usepackage{booktabs} 
\usepackage{amsmath,amsthm,enumitem,nicefrac,bbm,verbatim,amssymb,amsfonts,amscd, graphicx,algpseudocode,hyperref,float,mathtools,bm,xcolor,appendix,subcaption,thmtools,thm-restate,url}
\usepackage[utf8]{inputenc}
\usepackage[T1]{fontenc}
\usepackage[ruled, lined, linesnumbered, commentsnumbered, longend]{algorithm2e}
\usepackage{tikz}
\usepackage{tikz-3dplot}
\usetikzlibrary{positioning, decorations.pathreplacing, calc, intersections, pgfplots.groupplots}
\usetikzlibrary{positioning, decorations.markings}
\usepackage{pgfplots}
\usepackage{tikz-network}
\usetikzlibrary{shapes,decorations,arrows,calc,arrows.meta,fit,positioning}
\usepackage{titlesec}
\usepackage{thmtools}
\usepackage{thm-restate}
\pgfplotsset{compat=1.5}
\pgfplotsset{scaled x ticks=false}
\tikzset{
	-Latex,auto,node distance =1 cm and 1 cm,semithick,
	state/.style ={ellipse, draw, minimum width = 0.7 cm},
	point/.style = {circle, draw, inner sep=0.04cm,fill,node contents={}},
	bidirected/.style={Latex-Latex},
	el/.style = {inner sep=2pt, align=left, sloped}
}

\newtheorem{Theorem*}{Theorem}

\newtheorem{Claim*}[Theorem]{Claim}

\newtheorem{CounterExample*}{$\overline{\hbox{\bf Example}}$}

\newtheorem{Example*}[Theorem]{Example}

\newtheorem{Intuition*}[Theorem]{Intuition}
\newtheorem{Joke*}[Theorem]{Joke}

\newtheorem{Lemma*}[Theorem]{Lemma}
\newtheorem{Open problem}[Theorem]{Open problem}

\newtheorem{Question*}[Theorem]{Question}

\usepackage{fullpage}

\newtheorem{theorem}{Theorem}
\newtheorem{proposition}[theorem]{Proposition}

\newtheorem{lemma}[theorem]{Lemma}

\newtheorem{example}[theorem]{Example}
\newtheorem{definition}[theorem]{Definition}

\def \bSubexa    {\begin{subexa}}

\usepackage[colorinlistoftodos,textsize=scriptsize]{todonotes}

\newcommand{\ignore}[1]{}

\newcommand{\EE}{\mathbb{E}}

\newcommand{\RR}{\mathbb{R}}

\newcommand{\MM}{\mathbb{M}}

\def \cK     {{\cal K}}

\def \cP     {{\cal P}}

\def \cX     {{\cal X}}
\def \cY     {{\cal Y}}
\def \cZ     {{\cal Z}}

\def \Paren#1{{\left({#1}\right)}}

\def \Brack#1{{\left[{#1}\right]}}

\def\ignore#1{}

\newcommand{\bi}{\begin{itemize}}
\newcommand{\ei}{\end{itemize}}

\def\orpro{\mathop{\mathchoice
   {\vee\kern-.49em\raise.7ex\hbox{$\cdot$}\kern.4em}
   {\vee\kern-.45em\raise.63ex\hbox{$\cdot$}\kern.2em}
   {\vee\kern-.4em\raise.3ex\hbox{$\cdot$}\kern.1em}
   {\vee\kern-.35em\raise2.2ex\hbox{$\cdot$}\kern.1em}}\limits}

\def\andpro{\mathop{\mathchoice
 {\wedge\kern-.46em\lower.69ex\hbox{$\cdot$}\kern.3em}
 {\wedge\kern-.46em\lower.58ex\hbox{$\cdot$}\kern.25em}
 {\wedge\kern-.38em\lower.5ex\hbox{$\cdot$}\kern.1em}
 {\wedge\kern-.3em\lower.5ex\hbox{$\cdot$}\kern.1em}}\limits}

\def\simge{\mathrel{%
   \rlap{\raise 0.511ex \hbox{$>$}}{\lower 0.511ex \hbox{$\sim$}}}}

\def\simle{\mathrel{
   \rlap{\raise 0.511ex \hbox{$<$}}{\lower 0.511ex \hbox{$\sim$}}}}

\providecommand{\email}[1]{\href{mailto:#1}{\nolinkurl{#1}\xspace}}

\newcommand{\sex}{S}
\newcommand{\lsex}{s}
\newcommand{\dept}{D}
\newcommand{\ldept}{d}
\newcommand{\outcome}{A}
\newcommand{\loutcome}{a}
\newcommand{\model}{M}

\newcommand{\doop}[1]{\text{do}(#1)}

\newcommand{\modelsedge}{\MM_{\text{cf}+}}
\newcommand{\modelsedgeSD}{\MM_{\text{cf}+SD}}

\newcommand{\nullgraph}{H^0_{\text{cf-graph}+}}
\newcommand{\nullgraphrelax}{H^0_{\text{cf-graph}}}

\newcommand{\nullgraphunconf}{H^0_{\text{no-cf-graph}}}

\newcommand{\nullobsunconf}{H^0_{\text{no-cf-obs}}}

\newcommand{\nullinter}{H^0_{\text{cf-inter}+}}
\newcommand{\nullinterrelax}{H^0_{\text{cf-inter}}}

\newcommand{\nullctrfunconf}{H^0_{\text{no-cf-ctrf}}}
\newcommand{\nullctrf}{H^0_{\text{cf-ctrf}+}}
\newcommand{\nullnotion}{H^0_{\text{cf-notion}+}}
\newcommand{\nullnotionrelax}{H^0_{\text{cf-notion}}}
\newcommand{\nullctrfrelax}{H^0_{\text{cf-ctrf}}}

\newcommand{\nullinterunconf}{H^0_{\text{no-cf-inter}}}
\newcommand{\nullgraphresp}{\bar{H}^0_{\text{graph}}}
\newcommand{\nullinterresp}{\bar{H}^0_{\text{inter}}}

\newcommand{\enop}{V}
\newcommand{\exrv}{W}
\newcommand{\spc}{\mathcal{X}}
\newcommand{\response}{R}
\newcommand{\respfunc}{r}
\newcommand{\distiv}{\cP_{\text{IV}+}}
\newcommand{\distmodeliv}{\cP_{\modeliv}}
\newcommand{\distinter}{\cP_{\text{cf-inter}}}
\newcommand{\distgraph}{\cP_{\text{cf-graph}}}
\newcommand{\distctrf}{\cP_{\text{cf-ctrf}}}
\newcommand{\distinterunconf}{\cP_{\text{no-cf-inter}}}
\newcommand{\distobsunconf}{\cP_{\text{no-cf-obs}}}
\newcommand{\distgraphunconf}{\cP_{\text{no-cf-graph}}}
\newcommand{\distctrfunconf}{\cP_{\text{no-cf-ctrf}}}

\newcommand{\modelsunconfedge}{\MM_{\text{no-cf}}}

\newcommand{\cg}[1]{G\Paren{#1}}
\newcommand{\indep}{\perp\!\!\!\perp}
\newcommand{\modeliv}{\MM_{\text{IV}+}}
\newcommand{\modelivZX}{\MM_{\text{IV}+ZX}}

\newcommand{\distnotion}{\cP_{\text{cf-notion}}}

\newcommand{\modelivrelax}{\MM_{\text{IV}}}
\newcommand{\mkiv}{\mathcal{K}_{\text{IV}}}
\newcommand{\mkmodeliv}{\mathcal{K}_{\MM_{\text{IV}}}}
\newcommand{\mkinter}{\cK_{\text{cf-inter}}}
\newcommand{\mkgraph}{\cK_{\text{cf-graph}}}
\newcommand{\mkctrf}{\cK_{\text{cf-ctrf}}}
\newcommand{\modelsedgerelax}{\MM_{\text{cf}}}
\newcommand{\mknotion}{\cK_{\text{cf-notion}}}

\title{Revisiting the Berkeley Admissions data: Statistical Tests for Causal Hypotheses}

\author[1]{Sourbh Bhadane} 
\author[1]{Joris M. Mooij}
\author[1]{Philip Boeken}
\author[2]{Onno Zoeter}
\affil[1]{Korteweg-de Vries Institute for Mathematics, University of Amsterdam}
\affil[2]{Booking.com, Amsterdam, The Netherlands}
  
\begin{document}
\maketitle

\begin{abstract}
Reasoning about fairness through correlation-based notions is rife with pitfalls. The 1973 University of California, Berkeley graduate school admissions case from \citet{BickelHO75} is a classic example of one such pitfall, namely Simpson’s paradox. The discrepancy in admission rates among male and female applicants, in the aggregate data over all departments, vanishes when admission rates per department are examined. We reason about the Berkeley graduate school admissions case through a causal lens. In the process, we introduce a statistical test for causal hypothesis testing based on Pearl's instrumental-variable inequalities \citep{Pearl95}. We compare different causal notions of fairness that are based on graphical, counterfactual and interventional queries on the causal model, and develop statistical tests for these notions that use only observational data. We study the logical relations between notions, and show that while notions may not be equivalent, their corresponding statistical tests coincide for the case at hand. We believe that a thorough case-based causal analysis helps develop a more principled understanding of both causal hypothesis testing and fairness.
\end{abstract}

\section{Introduction}\label{sec:intro}
In the fall of $1973$, the Graduate Division of the University of California, Berkeley, made admission decisions for $12763$ applicants to its $101$ departments. The admission rate for $8442$ male applicants was approximately $44.2\%$ and for $4321$ female applicants was approximately $34.6 \%$. This disparity prompted \citet{BickelHO75} to investigate whether the Graduate Admissions Office discriminated on the basis of sex. The authors found that despite there being a statistically significant disparity in the aggregate data, when each department was examined, the per-department admission rates did not differ significantly between the sexes, thus making this case an instance of Simpson's paradox. The resolution was that the ``proportion of women applicants tends to be high in departments that are hard to get into and low in those that are easy to get into''. The disparity was therefore attributed to societal biases and the authors concluded that there was ``no pattern of discrimination on the part of the admissions committee.''  

In the fairness literature, the Berkeley graduate admissions case is a canonical example of Simpson's paradox, which illustrates the limitations of correlation-based fairness notions such as demographic parity, and therefore motivates the need for causal reasoning of fairness. \citet[Section 4.5.4]{Pearl09} analyzes the Berkeley example and frames the conclusion of \citet{BickelHO75} as discerning the direct effect of sex on admissions outcome by conditioning on the mediator, namely department choice. Most works in the fairness literature that mention the Berkeley example follow this analysis, which is predicated on the assumption that the causal model includes no latent confounders while the causal graph is akin to a simple mediation graph with sex being the treatment, department choice being the mediator and the admissions decision being the outcome. However, in both \citet{Pearl09} and \citet{PearlMackenzie18}, Pearl notes that merely conditioning on department choice might not always be appropriate. In particular, he cites a fascinating exchange between William Kruskal and Peter Bickel in \cite{FairleyMosteller77} where Kruskal objects to the analysis in \cite{BickelHO75} by pointing out that controlling for department leads to erroneous conclusions if there is a confounder that affects department choice and admissions outcome. To the best of our knowledge, subsequent works that mention the Berkeley example, do not address the latent confounder issue, including \citet{Pearl09} where the analysis assumes that the common causes are observed. Further, while there are multiple causal fairness notions proposed in the literature,\footnote{Some directly inspired by the Berkeley admissions case, for example the path-dependent counterfactual fairness notion in \citet[Appendix S4]{KusnerLRS17}.} the issue of statistical testing of these fairness notions has received little attention. 

In this work, we undertake a causal reasoning exercise centered around the Berkeley admissions case. We take the view that a causal analysis is predicated on causal modeling assumptions that define a family of causal models. A \textit{fairness notion} is either an observational, interventional, counterfactual or a graphical query on a causal model which, as a result, defines a subset of the aforementioned family of causal models, i.e., a fairness notion defines a \textit{causal hypothesis}. Given that we usually have only observational data at hand the question of fairness boils down to statistical testing of a causal hypothesis. Indeed, the Berkeley admissions data can be thought of as sampled from the joint distribution of the sex, department choice and admissions outcome.\footnote{Albeit possibly post-selection, which we don't address in this work.}

For the Berkeley admissions case, we consider multiple fairness notions based on graphical, counterfactual and interventional queries to the family of causal models defined by our causal modeling assumptions, which allows latent confounding between department choice and admissions outcome. For these notions, we develop new statistical tests. One of our key insights is that the graphical notion of fairness can be tested by using the instrumental-variable (IV) inequalities \citep{Pearl95}, thus making our proposed statistical test a new test for the IV inequalities. Conversely, \emph{any} statistical test for the IV inequalities can be used to test for fairness in settings that are analogous to the Berkeley case. In the process, we also prove a result of independent interest, namely the sharpness of the IV inequalities for the case where the instrument and the effect are binary, and the treatment takes any finite number of values. For the Berkeley example, while our proposed fairness notions correspond to different rungs of the causal hierarchy and are in general not equivalent, we show, rather surprisingly, that the tests are equivalent within the IV setting. Although our results are inspired by the Berkeley case, they can also be applied in other analogous settings, e.g.\ to investigate sex discrimination in awarding distinctions to PhD students \citep{Bol23}.

\subsection{Related Work}
The question of fairness in decision-making and predictive systems has received increased attention since the past few decades. See \citet{HutchinsonMitchell19,BarocasHN23} for an excellent historical  and technical overview, respectively. While attempts at formalizing fairness lead to correlation-based notions such as fairness through unawareness \citep{DworkHPRZ12}, demographic parity, equality of odds \citep{HardtPS16} etc., purely observational notions of fairness are at odds with each other \cite{Chouldechova17, KleinbergMR17} and are prone to erroneous conclusions. On the other hand, observational notions of fairness are readily translated to statistical tests.

Causal analysis tools such as counterfactuals and interventions provide a framework suitable for fairness. As a result, multiple general fairness notions based on counterfactuals were proposed.  \citet{KusnerLRS17} defined a counterfactual fairness notion that required invariance of the distribution of the decision in a given context, with respect to hypothetical changes in the protected attribute. \citet{NabiShpitser18} and \citet{ZhangWW17} consider path-specific effects. \citet{Chiappa19} proposes a path-specific counterfactual fairness notion and a related notion appears in the appendix of \citet{KusnerLRS17}. Another separate line of work seeks to explain observed disparity through causal discrimination mechanisms \citep{ZhangBareinboim18, plecko2022causal}.

The Berkeley graduate admissions case makes an appearance in multiple papers to motivate the need for causal fairness notions. \cite{KilbertusRPHJS17, plecko2022causal,KusnerLRS17,Chiappa19,BerkKT23} are a few among many works. In addition, the Berkeley example also serves as a motivation to introduce path-specific notions given the assumption that the direct effect of sex on admissions outcome is the only `unfair' path. Also, see \citet{BarocasHN23} for a critique of this common assumption. \citet{Pearl09} considers the Berkeley example at length and illustrates the objection to controlling for the mediator by positing an observed confounder. 

Despite the fact that most causal fairness works mention the Berkeley example, to the best of our knowledge, no previous work gives a definitive answer to the question of fairness for the Berkeley dataset under unobserved confounding. \citet{KilbertusBKWS20} discusses the impact of unmeasured confounding under restrictive parametric assumptions. \citet{ZhangBareinboim18, plecko2022causal} consider fairness models that allow for specific forms of unobserved confounding. \citet{SchroderFF24} build on this by providing sensitivity analysis on fairness of prediction models. However, the kinds of unobserved confounding that they allow affects the sensitive attribute which is different from the kind we allow for in the Berkeley dataset. 



\section{Preliminaries} \label{sec:prelim}
We outline a few definitions that follow the formal setup of \cite{BongersFPM21}.
\begin{definition}[Structural Causal Model (SCM)]
A \textbf{Structural Causal Model (SCM)} is a tuple $\model = \Paren{\enop,\exrv,\spc,f,P}$ where a) $\enop,\exrv$ are disjoint, finite index sets of \textbf{endogenous} and \textbf{exogenous} random variables respectively, b) $\spc = \prod_{i \in \enop \cup \exrv} \spc_i $ is the \textbf{domain} which is a product of standard measurable spaces $\spc_i$, c) for every $v \in \enop$, $f_v:\spc \mapsto \spc_{v}$ is a measurable function,  and $f = (f_v)_{v \in V}$ is called the \textbf{causal mechanism}; the equations $X_v = f_v(X)$ for every $v \in \enop$ are called the \textbf{structural equations}, and d) $P\Paren{\spc_{\exrv}} = \bigotimes_{w \in \exrv} P(X_w)$ is the \textbf{exogenous distribution} which is a product of probability distributions $P(X_w)$ on $\spc_{w}$.  
\end{definition}

\begin{definition}[Parent]
Let $\model = \Paren{\enop,\exrv,\spc,f,P}$ be an SCM. $k \in \enop \cup \exrv$ is a \textbf{parent} of $v \in \enop$ if and only if it is not the case that for all $x_{\enop \backslash k}, f_v(x_V,X_W)$ is constant in $x_k$ (if $k \in \enop$, resp. $X_k$ if $k \in \exrv$)  $P(\spc_{\exrv})$-a.s..
\end{definition}

\begin{definition}[Causal Graph]
             Let $\model = \Paren{\enop, \exrv, \spc, f,P}$ be an SCM. The \textbf{causal graph}, $\cg{\model}$, is a directed mixed graph with nodes $V$, directed edges $u \longrightarrow v$ if and only if $u \in \enop$ is a parent of $v \in \enop$, and bidirected edges $u \leftrightarrow v$ if and only if $ \,\exists w \in \exrv$ that is a parent of both $u,v \in \enop$. 
\end{definition}
For simplicity of exposition, we restrict attention to acyclic SCMs, i.e.\ SCMs whose causal graph is acyclic (contains no directed cycle $X \rightarrow \cdots \rightarrow Y \rightarrow X$).
\begin{definition}[Observational Distribution]
Given an acyclic SCM, $\model= \Paren{\enop,\exrv,\spc,f,P}$, the exogenous distribution, $P$ and the causal mechanism $f$ induce a probability distribution over the endogenous variables which is called the \textbf{observational distribution}, $P_{\model}(X_{\enop})$.
\end{definition}

\begin{definition}[Hard Intervention]
Given an acyclic SCM, $\model= \Paren{\enop,\exrv,\spc,f,P}$, an intervention target $\,T \subseteq \enop$, and an intervention value $x_T \in \cX_{T}$, the \textbf{intervened SCM} is defined as $\model_{\doop{X_T=x_T}} \triangleq \Paren{\enop,\exrv,\spc,(f_{V \backslash T},x_T),P}.$ Further, the observational distribution of the intervened SCM, $P_{\model_{\doop{X_T=x_T}}}$, is called an \textbf{interventional distribution}, and denoted by $P_{M}\Paren{X_V \mid \doop{X_T=x_T}}$.
\end{definition}


    

\begin{definition}[Potential Outcome]
  Let $\model = \Paren{\enop, \exrv, \spc, f,P}$ be an acyclic SCM, $C \subseteq \enop$ and $x_C \in \spc_{C}$. A random variable $X^{\doop{x_C}}$ taking values in $\spc_{\enop} \times \spc_{\exrv}$ is called a \textbf{potential outcome} of $\model$ under intervention $x_C$ if a) its $\exrv$-component has the exogenous distribution specified by $\model$, i.e., $X^{\doop{x_C}}_{\exrv} \sim P\Paren{\spc_{\exrv}}$, and b) it satisfies $X^{\doop{x_C}}_C = x_C\ \mathrm{a.s.}$ and
    \begin{equation*}
        X^{\doop{x_C}}_{\enop \backslash C} = f_{\enop \backslash C}\Paren{x_C,X^{\doop{x_C}}_{\enop \backslash C},X^{\doop{x_C}}_{\exrv} } \text{ a.s.}.
    \end{equation*}
If $\,C = \emptyset$, we write $X$ instead of $X^{\doop{x_{\emptyset}}}$. 
When dealing with multiple potential outcomes, we assume that they share the same $\exrv$-component:
for any $n$, for $C_1, C_2, \cdots, C_n \subseteq \enop$ and for $x_{C_1} \in \cX_{C_1}, x_{C_2} \in \cX_{C_2}, \cdots, x_{C_n} \in \cX_{C_n}$, 
  $$X_{\exrv}^{\doop{x_{C_1}}} = X_{\exrv}^{\doop{x_{C_2}}}= \cdots =X_{\exrv}^{\doop{x_{C_n}}} = X_W \text{ a.s.}.$$ 
\end{definition}

Instrumental variables (IVs) are used to estimate the causal effect of a treatment $X$ on the outcome $Y$ in the presence of latent confounding. 
\begin{definition}[Instrumental Variable (IV) Model Class]\label{def:IVmodelclass}
The \textbf{instrumental variable model class}, $\modelivrelax$, is a collection of SCMs $\model$ such that $\cg{\model}$ is a subgraph of Figure~\ref{fig:iv} where $Z$ is the \textbf{instrument}, $X$ is the \textbf{treatment}, and $Y$ is the \textbf{outcome}. 
\end{definition}
We make an explicit positivity assumption and define $\modeliv \triangleq \left \{ \model \in \modelivrelax:\forall z, P_{\model}(Z=z) > 0 \right\}$. Causal graphs of SCMs in the IV model class rule out the directed edge $Z \rightarrow Y$ and any latent confounding between $Z$ and $Y$ as well. In Section~\ref{app:SDconf} we show that our results for the IV model class hold even when confounding is allowed between $Z$ and $X$.

\begin{figure}[t]
     \centering
            \begin{tikzpicture}
            \tikzstyle{vertex}=[circle,fill=none,draw=black,minimum size=17pt,inner sep=0pt]
\node[vertex] (Z) at (0,0) {$Z$};
\node[vertex] (Y) at (3,0) {$Y$};
\node[vertex] (X) at (1.5,0) {$X$};
\path (Z) edge (X);
\path (X) edge (Y);
\path[bidirected] (X) edge[bend left=60] (Y);
            \end{tikzpicture}
        \caption{Causal graph of $M \in \modelivrelax$} 
        \label{fig:iv}
        \end{figure}

\section{Berkeley Case: No Latent Confounding}\label{sec:modeling}

\begin{figure*}[t]
     \centering
     \begin{subfigure}{0.32\linewidth}
     \centering
            \begin{tikzpicture}
            \tikzstyle{vertex}=[circle,fill=none,draw=black,minimum size=17pt,inner sep=0pt]
\node[vertex] (S) at (0,0) {$S$};
\node[vertex] (A) at (2,0) {$A$};
\node[vertex] (D) at (1,1) {$D$};
\path (S) edge (D);
\path (D) edge (A);
\path (S) edge (A);
            \end{tikzpicture}
        \caption{$\model \in \modelsunconfedge$}
        \label{fig:no-cf-edge}
\end{subfigure}
     \begin{subfigure}{0.32\linewidth}
     \centering
            \begin{tikzpicture}
            \tikzstyle{vertex}=[circle,fill=none,draw=black,minimum size=17pt,inner sep=0pt]
\node[vertex] (S) at (0,0) {$S$};
\node[vertex] (A) at (2,0) {$A$};
\node[vertex] (D) at (1,1) {$D$};
\path (S) edge (D);
\path (D) edge (A);
\path[bidirected] (D) edge[bend left=60] (A);
\path (S) edge (A);
            \end{tikzpicture}
        \caption{$\model \in \modelsedgerelax$} 
        \label{fig:cf-edge}
        \end{subfigure}
         \begin{subfigure}{0.32\linewidth}
     \centering
            \begin{tikzpicture}
            \tikzstyle{vertex}=[circle,fill=none,draw=black,minimum size=17pt,inner sep=0pt]
\node[vertex] (S) at (0,0) {$S$};
\node[vertex] (A) at (2,0) {$A$};
\node[vertex] (D) at (1,1) {$D$};
\path (S) edge (D);
\path (D) edge (A);
\path[bidirected] (D) edge[bend left=60] (A);
            \end{tikzpicture}
        \caption{$\model \in \nullgraph$ and $\model \in \modeliv$} 
        \label{fig:cf-edge-iv}
        \end{subfigure}
        \caption{Causal graphs, $\cg{\model}$, assumed in various model classes.}
\end{figure*}


Under the semantic framework of SCMs, we first make the same causal modeling assumptions that are commonplace in works that mention the Berkeley admissions case. We compare fairness notions that are tied to these modeling assumptions, with the view that modeling assumptions describe a family of SCMs and fairness notions define a subset of this family. We relate existing general notions of fairness in the literature to this viewpoint. While this is a re-examination of the various existing analyses of the Berkeley admissions case, in the next section, we relax the causal modeling assumptions and consider the more general family of models that allow for confounding between department choice and admissions outcome.

The set of endogenous variables consists of the protected attribute, namely sex of the applicant, $\sex$, the department they applied to, $\dept$, and the decision of the admissions committee, $\outcome$. We assume that $\sex, \outcome$ are binary variables and $\dept$ is a discrete-valued variable taking finite number of values, where $\sex=0,1$ corresponds to male, female applicants, respectively, and $\outcome=0,1$ corresponds to reject and accept, respectively.\footnote{The assumption of binary sex is purely for mathematical simplicity.} Given that, possibly, societal biases nudge applicants to departments at differing rates depending on their sex, we assume that $\sex$ affects $\dept$. Since departments are the primary decision-making units and have different admission rates, we also assume that $\dept$ affects $\outcome$. The question of whether acceptance decisions discriminate against sex centers around the \emph{direct} causal effect of $\sex$ on $\outcome$ (w.r.t.\ the three variables $\sex, \dept, \outcome$), and therefore we allow such an effect in the model. Indeed, if such a direct causal effect is absent, then the applicant's sex does not directly affect the admission outcome (only indirectly through department choice). On the other hand, if such a direct causal effect is present, then we would consider this unfair, provided that we know that it is not mediated purely via latent unprotected attributes (e.g., undergraduate department, or undergraduate university chosen by the applicant) that could play a similar `fair' mediating role as $\dept$. We will---for the time being---assume no latent unprotected attributes exist that mediate an effect of sex on admissions outcome, and come back to this point in Section~\ref{sec:discussion}.

In this section, we also assume the absence of any confounding between the variables (in addition to the absence of any selection bias). The structural equations are given by
\begin{align}\label{eq:no-cf-edge}
    \sex &= f_{\sex}(U_{\sex}) \nonumber, \\
    \dept &= f_{\dept}(\sex, U_{\dept}), \\
    \outcome &= f_{\outcome}(\sex,\dept,U_{\outcome}), \nonumber
\end{align}
where $U_{\sex},U_{\dept}$ and $U_{\outcome}$ denote independent exogenous random variables. We denote the family of SCMs parameterized by the functions in \eqref{eq:no-cf-edge} and the exogenous distribution as $\modelsunconfedge$. 
For $\model \in \modelsunconfedge$, the causal graph $G(\model)$ is a subgraph of the directed acyclic graph (DAG) in Figure~\ref{fig:no-cf-edge}.

\subsection{Fairness Notions}
We define a fairness notion to be a certain condition that is required to be satisfied by a causal model to be deemed fair. These conditions can take the form of constraints expressed in terms of observational, interventional, counterfactual or graphical queries on the SCMs in the families of causal models defined by modeling assumptions, in our case, $\modelsunconfedge$. While the criteria for the fairness notions in Section~\ref{sec:modeling} are phrased in terms of queries corresponding to different rungs of the causal ladder, in our case, any condition can only be tested using observational data.




The investigation of Berkeley's admission data was initiated on the observation that the well-known fairness notion of demographic parity $P_{\model}\Paren{\outcome=1 \mid \sex = 0} = P_{\model}\Paren{\outcome=1 \mid \sex = 1}$ did not hold. This fairness notion is based purely on observational data and we have already noted that it falls prey to Simpson's paradox. We now present another observational notion of fairness that can be seen as a conditional version of demographic parity.
\ifdefined\SINGLE
\begin{definition}[Observational Notion of Fairness]
$\model \in \modelsunconfedge$ is fair according to the observational notion of fairness if it belongs to the null hypothesis set 
\begin{align*}
\nullobsunconf \triangleq &\left \lbrace \model \in \modelsunconfedge :  \forall \ldept, \lsex, P_{\model}\Paren{\dept = \ldept, \sex = \lsex} >0 \right.\\
&\left. \implies P_{\model}\Paren{\outcome=1 \mid \sex =\lsex, \dept = \ldept} = P_{\model}\Paren{\outcome=1 \mid \dept = \ldept} \right \rbrace.
\end{align*}
\end{definition}
\else
\begin{definition}[Observational Notion of Fairness]
$\model \in \modelsunconfedge$ is fair according to the observational notion of fairness if it belongs to 
\begin{align*}
&\nullobsunconf \triangleq \left \lbrace \model \in \modelsunconfedge :  \forall \ldept, \lsex,  P_{\model}\Paren{\dept = \ldept, \sex = \lsex} >0 \right.\\
&\!\!\implies \!\!P_{\model}\Paren{\outcome=1 \,|\, \sex =\lsex, \dept = \ldept} = P_{\model}\Paren{\outcome=1 \,|\, \dept = \ldept} \rbrace.
\end{align*}
\end{definition}
\fi

\citet{BickelHO75} proposed this notion for the Berkeley data. It is equivalent to the conditional independence $\outcome \indep \sex \mid \dept$.
Hence, any valid conditional independence test for $\outcome \indep \sex \mid \dept$ provides a statistical test for this notion. Indeed, the analysis of \citet{BickelHO75} shows that the data contain not enough evidence to reject the null hypothesis that this conditional independence holds, and therefore, concludes fairness. 

From the causal graph of $\modelsunconfedge$ in Figure~\ref{fig:no-cf-edge}, a natural subset of fair causal models is those without the edge $\sex \rightarrow \outcome$. 

\begin{definition}[Graphical Notion of Fairness]\label{def:graph_fairness}
     $\model \in \modelsunconfedge$ is fair according to the graphical notion of fairness if it belongs to the null hypothesis set $\nullgraphunconf \triangleq \left \lbrace \model \in \modelsunconfedge : \sex \rightarrow \outcome \notin \cg{\model} \right \rbrace$.
\end{definition}


\citet[Section 4.5.3]{Pearl09} discusses the direct effect in the context of the Berkeley admissions example, where he objects to conditioning on the department and instead proposes intervening on department choice, which corresponds to the controlled direct effect (CDE) \citep{Pearl01} of the `treatment', $\sex$, on the outcome, $\outcome$, for every value of the mediator, i.e., every department choice $\ldept$. 

\ifdefined\SINGLE
\begin{definition}[Interventional Notion of Fairness]
    $\model \in \modelsunconfedge$ is fair according to the interventional notion of fairness if it belongs to 
     \begin{equation*}\label{eq:interfairunconf}
    \nullinterunconf \triangleq \left \lbrace \model \in \modelsunconfedge: \forall \ldept,\lsex, P_{\model}\Paren{\outcome=1\mid\doop{\sex=\lsex},\doop{\dept=\ldept}} = P_{\model}\Paren{\outcome=1\mid\doop{\dept = \ldept}} \right \rbrace.
    \end{equation*}
    \end{definition}
\else
\begin{definition}[Interventional Notion of Fairness]
    $\model \in \modelsunconfedge$ is fair according to the interventional notion of fairness if it belongs to the null hypothesis set 
     \begin{align*}\label{eq:interfairunconf}
    &\nullinterunconf \triangleq \left \lbrace \model \in \modelsunconfedge: \forall \ldept,\lsex \right. \nonumber\\
   & \left. P_{\model}\Paren{\outcome=1\mid\doop{\sex=\lsex},\doop{\dept=\ldept}} \right. \nonumber\\
   &\left. = P_{\model}\Paren{\outcome=1\mid\doop{\dept = \ldept}} \right \rbrace.
    \end{align*}
\end{definition}
\fi

Recent analyses of the Berkeley example emphasize counterfactual notions of fairness. In \citet[Section 4.5.4]{Pearl09}, \citet{PearlMackenzie18}, Pearl considers a counterfactual quantity, namely the natural direct effect (NDE) \citep{RobinsG92, Pearl01}  by motivating a hypothetical experiment where ``all female candidates retain their department preferences but change their gender [sex] identification (on the application form) from female to male''. Subsequent causal fairness works \citep{NabiShpitser18, Chiappa19} build on this and propose fairness notions based on known path-specific versions of NDE where the `direct path' from $\sex$ to $\outcome$ is viewed as `unfair' as opposed to the `fair' path $\sex \rightarrow \dept \rightarrow \outcome$. For the Berkeley example, the NDE$(s\rightarrow s')$ is given by 
\begin{equation*}
P_{\model}\Paren{\outcome^{\doop{\sex = \lsex', \dept = \dept^{\doop{\sex=\lsex}}}}=1} - P_{\model}\Paren{\outcome^{\doop{\sex=s}}=1}
\end{equation*}
for $\lsex \neq \lsex'$. Note that by Pearl's mediation formula \citep{Pearl01}, the above is identified (assuming $\forall \ldept, \lsex, P_{\model}\Paren{\dept = \ldept, \sex = \lsex} > 0$) as 
\ifdefined\SINGLE
\begin{equation*}
    \sum\limits_{\ldept} \left( P_{\model}\Paren{\outcome=1 \mid \dept = \ldept, \sex = \lsex'}  - P_{\model}\Paren{\outcome=1 \mid \dept = \ldept, \sex = \lsex}\right)P_{\model}\Paren{\dept = \ldept \mid \sex= \lsex}.
\end{equation*}
\else
\begin{align*}
    \sum\limits_{\ldept} & \left( P_{\model}\Paren{\outcome=1 \mid \dept = \ldept, \sex = \lsex'} \right. \\
    &\left. - P_{\model}\Paren{\outcome=1 \mid \dept = \ldept, \sex = \lsex}\right)P_{\model}\Paren{\dept = \ldept \mid \sex= \lsex}.
\end{align*}
\fi
This implies that if the observational notion of fairness and positivity hold, the NDE is $0$. However, the converse is not necessarily true. For example, if one department favors male applicants and another favors female applicants, then the NDE could be $0$ while it is not necessary that the observational notion of fairness holds. 

Other counterfactual notions of fairness include those by \citet{KusnerLRS17}. The authors define a counterfactual fairness notion that implies demographic parity (see Section~\ref{app:kusnerctrfdemo} for a proof) for the Berkeley example; we have already seen that this particular fairness notion falls prey to Simpson's paradox. In the appendix, however, they define a path-dependent notion of counterfactual fairness.\footnote{This notion is specifically motivated by the Berkeley example.} In Section~\ref{app:kusnerpathnocf} we show that, in our setting, testing for the path-dependent counterfactual fairness notion is equivalent to testing for the conditional independence $A \indep S \mid D$. We now propose an alternate counterfactual notion of fairness and later compare testing of the same.


\begin{definition}[Counterfactual Notion of Fairness]\label{def:ctrf-nocf}
$\model \in \modelsunconfedge$ is fair according to the counterfactual notion of fairness if it belongs to the null hypothesis set
\ifdefined\SINGLE
\begin{equation*}\label{eq:nullctrf}
    \nullctrfunconf \triangleq \left \lbrace \model \in \modelsunconfedge: \forall \ldept, \lsex, P_M(A^{\doop{S=s,D=d}} = A^{\doop{D=d}}) = 1 \right \rbrace.
\end{equation*}
\else
\begin{align*}\label{eq:nullctrf}
    \nullctrfunconf &\triangleq \left \lbrace \model \in \modelsunconfedge:\forall \ldept, \lsex,  \right. \nonumber\\
    &\left. P_M(A^{\doop{S=s,D=d}} = A^{\doop{D=d}}) = 1 \right \rbrace
\end{align*}
\fi
\end{definition}
The alternative hypotheses are given by the complement of the null hypotheses w.r.t. $\modelsunconfedge$. We consider the interventional and counterfactual fairness notions as analogues of the absence of direct effect that are framed in terms of rung-$2$
 and rung-$3$ notions, namely interventional and counterfactual distributions, respectively. Therefore, the interventional notion of fairness considers comparing interventional distributions that result from intervening on all the variables with intervening on all variables but the protected variable, and the counterfactual fairness notions compare potential outcomes resulting from the same set of interventions. Given that the notions are defined on different rungs of the causal hierarchy, it is perhaps not surprising that they are nested accordingly. The assumption of no confounding simplifies the relations as we can prove equivalence of a few notions under positivity. The proof is deferred to Section~\ref{app:nested-nocf}.
\begin{restatable}{lemma}{unconfnested}\label{lem:notion_equiv}
\begin{equation*}\nullgraphunconf = \nullctrfunconf \subset \nullinterunconf \subset \nullobsunconf.\end{equation*}
If for all $s,d$, $P_{\model}(s,d) > 0$, then in addition, we have $\nullinterunconf = \nullobsunconf$.
\end{restatable}

Despite the nested nature of the fairness notions at different rungs of the causal hierarchy, we prove that the sets of observational distributions that these notions induce are identical. The proof is in Section~\ref{app:equiv-nocf}.
\begin{restatable}{theorem}{unconfequiv}\label{thm:unconf_test_equiv}
Let
\begin{align*} 
\distgraphunconf &\triangleq \left \lbrace P_{\model}\Paren{\dept,\outcome, \sex} : \model \in \nullgraphunconf \right \rbrace, \\
\distctrfunconf &\triangleq \left \lbrace P_{\model}\Paren{\dept,\outcome, \sex} : \model \in \nullctrfunconf \right \rbrace, \\
\distinterunconf &\triangleq \left \lbrace P_{\model}\Paren{\dept,\outcome, \sex} : \model \in \nullinterunconf \right \rbrace, \\
\distobsunconf &\triangleq \left \lbrace P_{\model}\Paren{\dept,\outcome, \sex} : \model \in \nullobsunconf \right \rbrace.
\end{align*}
Then $\distgraphunconf = \distctrfunconf = \distinterunconf = \distobsunconf.$
\end{restatable}


In summary, despite the fact that we analyze the Berkeley admissions case using multiple fairness notions, under the assumption of no confounding, with observational data, they can all be tested using a conditional independence test. 

If the data contains enough evidence to reject conditional independence, then the data generating mechanism is unfair w.r.t.\ the observational notion of fairness (assuming no latent unprotected mediators). On the other hand, if the data does not contain enough evidence to reject conditional independence, then the data generating mechanism is fair w.r.t.\ the observational notion of fairness. However, this extrapolation of the outcome of the statistical test on the fairness implications does not hold for the interventional, counterfactual and graphical notions. The following example illustrates that for the graphical notion of fairness, an unfaithful causal model, where $\outcome$ is directly affected by $\sex$, could satisfy conditional independence. 
\begin{example}\label{ex:unconfexample}
  Let $\model \in \modelsunconfedge$ be defined as $U_{\sex} \sim \text{Ber}(\delta), U_{\outcome} \sim \text{Ber}(\varepsilon), U_{\dept} \sim \text{Ber}(\frac{1}{2})$ where $\delta, \varepsilon \in [0,1]$ and $\sex = U_{\sex}, \dept=\sex \oplus U_{\dept}, \outcome = \sex \oplus \dept \oplus U_{\outcome}$. Here, $\outcome \indep \sex \mid \dept$ but $\sex$ is a parent of $\outcome$, i.e., $\model \in \nullobsunconf$, but $\model \notin \nullgraphunconf$. 
\end{example}
For the interventional notion of fairness, the following example illustrates that a causal model that violates positivity could satisfy conditional independence but not the interventional notion of fairness.
\begin{example}\label{ex:posexample}
    Let $\model \in \modelsunconfedge$ be defined as $U_{\sex}=0, U_{\dept}=0,U_{\outcome} \sim \text{Ber}\Paren{\varepsilon}$ where $\varepsilon \in [0,\frac{1}{2})$, and $\sex=0,\dept=0,\outcome=\sex \oplus U_A$. Here $\outcome \indep \sex \mid \dept$, but 
 for all $d$, $P_{\model}\Paren{\outcome = 1 \mid \doop{\sex=1}, \doop{D=d}} = 1-\varepsilon \neq P_{\model}\Paren{\outcome = 1 \mid \doop{D=d}} = \varepsilon$. Therefore, $\model \in \nullobsunconf$, but $\model \notin \nullinterunconf$.
\end{example}

So, if the outcome of the test is that conditional independence cannot be rejected ($\model \in \nullobsunconf$), then due to the aforementioned observations, we cannot conclude that the underlying causal model belongs to the causal null hypothesis of the interventional or counterfactual or graphical fairness notions, i.e., our conclusion is that fairness is ``undecidable'' with respect to these notions. However, if the outcome of the statistical test is that there is enough evidence in the data to reject conditional independence ($\model \notin \nullobsunconf$), then we can conclude that the underlying causal model does not belong to the causal null hypothesis of \textit{any} of the fairness notions. If we rule out the existence of any latent, unprotected mediator between $S$ and $A$, we can conclude that the data generating mechanism is unfair. 

In the next section, we enlarge the class of models to allow for confounding between $\dept$ and $\outcome$ and perform a similar reasoning exercise. 

\section{Berkeley Case: With Latent Confounding Between Department And Outcome}\label{sec:confounder}
\begin{figure}[!th]
    \centering
    \begin{tikzpicture}
\tikzstyle{vertex}=[circle,fill=none,draw=black,minimum size=17pt,inner sep=0pt]
\node[vertex] (S) at (0,0) {$S$};
\node[vertex] (A) at (2,0) {$A$};
\node[vertex] (D) at (1,1) {$D$};
\path (S) edge[bend left=10] (D);
\path (D) edge[bend left=10] (S);
\path (D) edge[bend left=10] (A);
\path (A) edge[bend left=10] (D);
\path[bidirected] (D) edge[bend left=60] (A);
\path[bidirected] (S) edge[bend left=60] (D);
\path[bidirected] (S) edge[bend right=60] (A);
\path (S) edge[bend left=10] (A);
\path (A) edge[bend left=10] (S);
    \end{tikzpicture}
    \caption{Causal graph of a model without assumptions}
    \label{fig:causal_modeling}
\end{figure}

We now take a more careful causal modeling approach. Instead of starting from variables and reasoning about structural equations that we allow, we start with assuming that all structural equations exist.\footnote{Since this allows for causal cycles, this would require using e.g. the framework of simple SCMs \citep{BongersFPM21}.} For the Berkeley example, Figure~\ref{fig:causal_modeling} shows a causal graph of an SCM that we start with. We now provide rationale for ruling out few structural equations. Based on chronology of events, we rule out those where $D$ directly affects $S$, where $A$ directly affects $D$ and where $A$ directly affects $S$. We rule out unobserved common causes of $S$ and $D$, and $S$ and $A$ since we model $S$ to be sex at birth. While latent selection bias might introduce bidirected edges \citep{ChenZM24} that are incident on $S$, we assume for now that there is no selection bias in the dataset. The resulting class of SCMs has structural equations of the form
\ifdefined\SINGLE
\begin{align}\label{eq:cf-edge}
    \sex &= f_{\sex}(U_{\sex}) \nonumber, \\
    \dept &= f_{\dept}(\sex,U, U_{\dept}), \\
    \outcome &= f_{\outcome}(\sex,\dept,U,U_{\outcome}), \nonumber
\end{align}
\else 
$\sex = f_{\sex}(U_{\sex}), \dept = f_{\dept}(\sex,U, U_{\dept}), \outcome = f_{\outcome}(\sex,\dept,U,U_{\outcome})$
\fi 
where $U,U_{\sex},U_{\dept}$ and $U_{\outcome}$ denote independent exogenous random variables. 
We define $\modelsedgerelax$ to be the family of models parameterized by the above structural equations and the exogenous distribution. Further, we define $\modelsedge = \left\{ \model \in \modelsedgerelax : \forall s, P_{\model}(S=s) > 0 \right\}$. For $\model \in \modelsedgerelax$ (and $\modelsedge$), the causal graph is a subgraph of the one shown in Figure~\ref{fig:cf-edge}. 

Although we arrived at allowing confounding between department and outcome through a careful causal modeling approach, this is not a novel consideration. In particular, Kruskal \citep[Pg 128-129]{FairleyMosteller77} demonstrated an example where the existence of a latent confounder, such as state of residence, can render \citet{BickelHO75}'s analysis incorrect. Other natural latent confounders include, for example, level of department-specific technical skills that influence both the department choice of an applicant and the admissions outcome.

Since our modeling assumptions expand the family of SCMs under consideration to $\modelsedge$, the fairness notions that we discussed in the previous section are modified accordingly to obtain null hypothesis sets $ \nullgraph, \nullinter$ and $\nullctrf$. 
\ifdefined\SINGLE
\begin{definition}[Fairness Notions with Confounding]\label{def:notions-cf}
For $\model \in \modelsedge$ the null hypothesis set corresponding to the interventional, counterfactual and graphical notion of fairness are 
\begin{align*}\label{eq:cf-def}
    \nullinter &\triangleq \left \lbrace \model \in \modelsedge: \forall \ldept,\lsex, P_{\model}\Paren{\outcome=1\mid\doop{\sex=\lsex},\doop{\dept=\ldept}} = P_{\model}\Paren{\outcome=1\mid\doop{\dept = \ldept}} \right \rbrace, \\
     \nullctrf &\triangleq \left \lbrace \model \in \modelsedge: \forall \ldept, \lsex, P_M(A^{\doop{S=s,D=d}} = A^{\doop{D=d}}) = 1 \right \rbrace, \\
     \nullgraph&\triangleq \left \lbrace \model \in \modelsedge : \sex \rightarrow \outcome \notin \cg{\model} \right \rbrace.
\end{align*}
\end{definition}
\else
\begin{definition}[Fairness Notions with Confounding]\label{def:notions-cf}
For $\model \in \modelsedge$ the null hypothesis set corresponding to the interventional, counterfactual and graphical notion of fairness are 
\begin{align*}\label{eq:cf-def}
    \nullinter &\triangleq \left \lbrace \model \in \modelsedge: \forall \ldept,\lsex, \right.\\
    &\left. P_{\model}\Paren{\outcome=1\mid\doop{\sex=\lsex},\doop{\dept=\ldept}} \right.\\
    &\left. = P_{\model}\Paren{\outcome=1\mid\doop{\dept = \ldept}} \right \rbrace, \\
     \nullctrf &\triangleq \left \lbrace \model \in \modelsedge: \forall \ldept, \lsex, \right.\\
     &\left. P_M(A^{\doop{S=s,D=d}} = A^{\doop{D=d}}) = 1 \right \rbrace, \\
     \nullgraph &\triangleq \left \lbrace \model \in \modelsedge : \sex \rightarrow \outcome \notin \cg{\model} \right \rbrace.
\end{align*}
\end{definition}
\fi
While the above notions generalize straightforwardly from the no-confounder setting, this is no longer the case for the observational notion. In addition, while the statistical tests for the no-confounder model are straightforward, this is no longer the case for the aforementioned null hypotheses since $\outcome \not\!\perp\!\!\!\perp \sex \mid \dept$ in general. 

\subsection{Graphical Notion and the Instrumental Variable (IV) Inequalities}\label{subsec:graph_iv}

In the presence of latent confounding, graphical queries, such as absence of edges, impose equality or inequality constraints \citep{Evans16, WolfeSF19} in addition to conditional independence constraints which are the only constraints imposed by a DAG. For the Berkeley case with confounding, since the path $\sex \rightarrow \dept \leftrightarrow \outcome$ is open when conditioned on $\dept$, we have $\sex \not\!\perp\!\!\!\perp \outcome \mid \dept$ in general. 
Our test for the graphical notion of fairness for $\modelsedge$ stems from the observation that a model $\model \in \nullgraph$ lies in the instrumental variable (IV) model class $\modeliv$ where $\sex$ is considered the instrument, $\dept$ the treatment, and $\outcome$ the effect. If all modeled endogenous variables are discrete-valued, a necessary condition for the observational distribution\footnote{While we express the IV inequalities as a condition satisfied by the observational distribution, in Section~\ref{app:iv} we reason that they are more appropriately expressed as conditions in terms of $P_{\model}(X,Y \mid \doop{Z})$.} resulting from $\model \in \modeliv$ is to satisfy the IV inequalities \citep{Pearl95}, which in the context of Figure~\ref{fig:iv} are given by  
\begin{equation}\label{eq:iv}
    \max_{x} \sum_{y} \max_{z} P_{\model}\Paren{X=x,Y=y\mid Z=z} \leq 1. 
\end{equation}
Since the IV inequalities are only necessary conditions, an arbitrary distribution on $X,Y,Z$ that satisfies the IV inequality does not necessarily imply that it is an entailed distribution of a model from the IV model class. \citet{Bonet01} showed that for the binary instrument, treatment and effect case, the IV inequalities are also sufficient conditions. In Theorem~\ref{thm:iv_tight}, we show that for the case where the instrument and outcome are binary and the treatment is discrete-valued with finite support, any distribution that satisfies the IV inequality is also entailed by some causal model from the IV model class. To the best of our knowledge, Theorem~\ref{thm:iv_tight} is a novel result. We defer the proof to Section~\ref{app:ivsharp}.
\ifdefined\SINGLE
\begin{restatable}{theorem}{ivtight}\label{thm:iv_tight}
Let $X,Y,Z$ be discrete random variables defined on $\cX,\cY,\cZ$ respectively, with $|\cX| = n\geq 2, |\cY|=2, |\cZ| =2$. Let the set of joint distributions that satisfy the IV inequalities be defined as $\distiv \triangleq \left \lbrace P(X,Y,Z) : P(X,Y \mid Z)\text{ satisfies }\eqref{eq:iv} \text{ and } \forall z, P(Z=z)>0 \right \rbrace$. Define $\distmodeliv \triangleq \left \lbrace P_{\model}(X,Y,Z) : \model \in \modeliv \right \rbrace.$ Then $\distiv = \distmodeliv.$
\end{restatable}
\else
\begin{restatable}{theorem}{ivtight}\label{thm:iv_tight}
Let $X,Y,Z$ be discrete random variables defined on $\cX,\cY,\cZ$ respectively, with $|\cX| = n\geq 2, |\cY|=2, |\cZ| =2$. Define $\distmodeliv \triangleq \left \lbrace P_{\model}(X,Y,Z) : \model \in \modeliv \right \rbrace$ and the set of joint distributions that satisfy the IV inequalities as \begin{align*}\distiv &\triangleq \left \lbrace P(X,Y,Z) : P(X,Y \mid Z)\text{ satisfies }\eqref{eq:iv} \right.\\
&\left.\text{ and } \forall z, P(Z=z)>0 \right \rbrace.\end{align*}  Then $\distiv = \distmodeliv.$
\end{restatable}
\fi
For the Berkeley admissions case,
the observational distribution satisfying the IV inequalities implies that there exists a causal explanation (model) where the directed edge $\sex \rightarrow \outcome$ is absent, i.e., given that $P(\outcome,\dept, \sex) \in \distiv$, there exists $\model \in \modeliv$ such that $P_{\model}(\outcome,\dept,\sex) = P(\outcome,\dept,\sex)$. On the other hand, the observational distribution violating the IV inequalities does not necessarily imply that the edge $\sex \rightarrow \outcome$ is present since the IV model class, $\modeliv$, is only a subset of all the models that do not contain the edge $\sex \rightarrow \outcome$ in the causal graph. For example, the existence of latent confounding between $\sex$ and $\outcome$ in a model $\model$ may result in $\model \notin \modeliv$, even though $\cg{\model}$ does not necessarily contain the directed edge $\sex \rightarrow \outcome$. However, the causal modeling assumption that defined $\modelsedge$ rules out latent confounding between $\sex$ and $\outcome$. Therefore, given our modeling assumptions, $\nullgraph = \modeliv$, and in turn, we conclude that violating the IV inequalities implies that $\model \in \modelsedge \backslash \nullgraph$. As in the previous section, it is possible that causal models that lie outside $\nullgraph$  (``unfair'' models) induce observational distributions that lie in $\distiv$, i.e., satisfy the IV inequalities. Therefore, satisfying the IV inequalities is not conclusive evidence that the data-generating mechanism is fair, i.e., our conclusion should be that fairness is undecidable. In Section~\ref{sec:bayesiantest} we introduce a Bayesian test for the IV inequalities.

\subsection{Bounds on Interventional Notion of Fairness}\label{subsec:bounds}
For $\model \in \modelsedge$, the interventional notion of fairness is equivalent to a vanishing CDE, where the latter is not identifiable in our case.
By a response-function parameterization \citep{Balke95, BalkePearl97} of $\model \in \modelsedge$, we can express the interventional distributions in Definition~\ref{def:notions-cf} as a linear function of response variables. Further, the observational distribution is also expressed as a linear function of the response variables. Using the symbolic linear programming approach of \citet{Balke95}, we obtain upper and lower bounds in terms of the observational distribution, specifically, $P_{\model}\Paren{\outcome,\dept \mid \sex}$. Indeed, \citet{CaiKPT08} express the same bounds which we reproduce below. The CDE
\ifdefined\SINGLE
\begin{equation*}P_{\model}\Paren{\outcome=1\mid \doop{\sex=1}, \doop{\dept=\ldept}} - P_{\model}\Paren{\outcome=1\mid \doop{\sex=0}, \doop{\dept=\ldept}},
\end{equation*}
\else
\begin{align*}
&P_{\model}\Paren{\outcome=1\mid \doop{\sex=1}, \doop{\dept=\ldept}}\\
&- P_{\model}\Paren{\outcome=1\mid \doop{\sex=0}, \doop{\dept=\ldept}},
\end{align*}
\fi
lies in the interval
\ifdefined\SINGLE
\begin{align*}
    &\left[ \Pr\left(\outcome=1,\dept=\ldept\mid\sex=1\right) + \Pr\left(\outcome=0,\dept=\ldept\mid \sex=0\right) - 1,\right.\\
    & \left.1 - \Pr\left(\outcome=0,\dept=\ldept\mid\sex=1\right) - \Pr\left(\outcome=1,\dept=\ldept\mid\sex=0\right)\right].
\end{align*}
\else
\begin{align*}
    &\left[ \Pr\left(\outcome=1,\ldept\mid \sex=1\right) + \Pr\left(\outcome=0,\ldept\mid\sex=0\right) - 1,\right.\\
    & \left.1 - \Pr\left(\outcome=0,\ldept\mid\sex=1\right) -\Pr\left(\outcome=1,\ldept\mid\sex=0\right)\right].
\end{align*}
\fi 
For the interventional notion of fairness, the CDE must be $0$ for all $\ldept$. By setting the lower bound to be at most $0$ and the upper bound to be at least $0$, we recover the IV inequalities in \eqref{eq:iv}. While \citet{CaiKPT08} do not point out the connection to the IV inequalities, they find it ``remarkable that we [they] get such a simple formula, consisting of only one additive expression in the lower bound and one additive expression in the upper bound''. In the next subsection, we show that the connection to the IV inequalities is not a coincidence.   

\subsection{A Family of Equivalent Tests}\label{subsec:equivalence}


The graphical and interventional fairness notions end up imposing identical constraints on the observational distribution.
However, note that $\nullinter \supseteq \nullgraph$. In fact, we prove in Section~\ref{app:nested-cf} that $ \nullgraph = \nullctrf \subset \nullinter$. Models in $\nullinter \backslash \nullgraph$ (Example~\ref{ex:unconfexample} with $\varepsilon=\frac{1}{2}$) are such that the edge $\sex \rightarrow \outcome$ exists in the causal graph and yet, the interventional queries in Definition~\ref{def:notions-cf} are equal.
Given that the null hypothesis of the interventional fairness notion is a strict superset of that of the graphical fairness notion, we might expect the same relation to hold in the resulting set of observational distributions for these hypotheses, thus giving us potentially different tests. In contrast, like in Section~\ref{sec:modeling}, we show that the corresponding sets of observational distributions resulting from models in $\nullinter$, $\nullctrf$, $\nullgraph$ are identical. Section~\ref{app:equiv-cf} contains the proof.

\begin{restatable}{theorem}{confequiv}\label{thm:equivalence}
Let 
\begin{align*}
\distgraph &\triangleq \left \lbrace P_{\model}\Paren{\dept,\outcome, \sex} : \model \in \nullgraph \right \rbrace, \\
\distinter &\triangleq \left \lbrace P_{\model}\Paren{\dept,\outcome, \sex} : \model \in \nullinter \right \rbrace, \\
\distctrf &\triangleq \left \lbrace P_{\model}\Paren{\dept,\outcome, \sex} : \model \in \nullctrf \right \rbrace.
\end{align*}
Then $\distinter = \distctrf = \distgraph = \distiv,$ where $\distiv$ is defined in Theorem~\ref{thm:iv_tight}.
\end{restatable}

In summary, testing for the graphical, interventional and counterfactual notions of fairness, with confounding, all boil down to testing the IV inequalities.

\subsection{Comparison With Existing Fairness Notions}

The utility of considering statistical tests is that we can now compare different fairness notions for a particular case with respect to the same causal modeling assumptions. In this section, we consider the three existing counterfactual fairness notions, namely the NDE \citep{NabiShpitser18, Chiappa19}, and the counterfactual and path-dependent counterfactual fairness notions in \citet{KusnerLRS17}. 

For NDE, \citet{KaufmanKMGP05} obtain bounds for the all-binary setting. Using these bounds, we obtain a strictly weaker test than the IV inequalities.\footnote{Intuitively, the reason is the same as in Section 3; the NDE averages over departments, and a positive bias in one department may cancel out against a similarly strong negative bias in another department. Hence, vanishing NDE does not imply that each department takes fair decisions.} We show in Section~\ref{app:ctrfkusner-cf} that the counterfactual notion of fairness of \citet{KusnerLRS17} implies demographic parity even when confounding is allowed. In Section~\ref{app:pathwise-cf} we show that testing the path-dependent counterfactual fairness notion of \citet{KusnerLRS17} is equivalent to testing the IV inequalities.

\section{Bayesian Testing Procedure} \label{sec:bayesiantest}

We start with proposing a Bayesian test for IV inequalities using finite data. Although  \citet{Ramsahai08} proposed a frequentist test and derived the distribution of the likelihood ratio, providing a means to obtain a p-value for the case of a binary treatment, it is unclear what the test would be for our case where the ``treatment'' is not binary. In addition, \citet{WangRR17} provide a frequentist test that involves multiple one-sided independence tests. In contrast, the Bayesian test has straightforward extensions to other hypotheses that are expressed in terms of the observational distribution, including bounds on unidentifiable causal queries.

Consider discrete random variables $X,Y,Z$ where $|\cX|=n, |\cY|=2, |\cZ|=2$. The observational distribution $P(X,Y,Z)$ lies in the $4n-1$ dimensional simplex, denoted by $\Delta$, which is a subset of $\RR^{4n}$. We consider a Bayesian model selection procedure where 
\begin{align*}
    \MM_0 &= \left \lbrace \theta \in \Delta : \theta \text{ satisfies the IV inequalities} \right \rbrace, \\
    \MM_1 &= \left \lbrace \theta \in \Delta : \theta \text{ does not satisfy the IV inequalities} \right \rbrace.
\end{align*}
Note that $\MM_0 \,\dot{\cup}\, \MM_1 = \Delta$. Given a finite dataset of $(X,Y,Z)$ tuples, denoted by $R_1, R_2, \cdots, R_m$, and choices of prior distributions for the models, $\pi\Paren{\theta\mid\MM_0},\pi\Paren{\theta\mid\MM_1}$, we report a confidence interval for the posterior probability of satisfying the IV inequalities, i.e., $P\Paren{\MM_0\mid R_1,R_2, \cdots, R_m} = \int_{\theta \in \MM_0} P\Paren{\theta \mid R_1,R_2, \cdots, R_m} d\theta$. Given the posterior density $P\Paren{\theta \mid R_1,R_2, \cdots, R_m}$, we estimate the posterior probability of $\MM_0$ by IID sampling from the posterior density $n$ times and counting how often the sample satisfies the IV inequalities, which we denote by $N$. Since $N$ is a binomial random variable with parameters $n$ and $P\Paren{\MM_0\mid R_1,R_2, \cdots ,R_m}$, a confidence interval on $P\Paren{\MM_0 \mid R_1,R_2, \cdots R_m}$ is readily obtained by the Clopper-Pearson method \citep{ClopperPearson34}.

\noindent\textbf{Results on Berkeley admission data: }We use the \texttt{UCBAdmissions} \citep{UCBAdmissions} dataset from $\texttt{R}$ that contains counts for each sex-department-admissions outcome tuple for the $6$ largest departments. Therefore, $|\cX| = 6$, $|\cY|=2, |\cZ|=2$. Since the data satisfies the positivity for sex, we can use the Bayesian model selection procedure. For parameters \ifdefined\SINGLE $$\theta = \Paren{ P(d,a,s): s \in \cX_S, d \in \cX_D, a \in \cX_A},$$ \else $\theta = \Paren{ P(d,a,s): s \in \cX_S, d \in \cX_D, a \in \cX_A}$,\fi we choose a flat Dirichlet prior over $\Delta$ giving us $\pi(\theta | \MM_i)= c_i \text{Dir}\Paren{1,1,\cdots,1}\bm{1}\left[ \theta \in \MM_i \right]$ where $c_i$ is a normalizing constant. The counts from the data are used to obtain the posterior, $P(\theta \mid R_1, R_2, \cdots, R_m)$ which is also a truncated Dirichlet distribution. Using $n=10^6$ samples, we observe no violations of the IV inequalities. Therefore, the confidence interval for the posterior probability of the Berkeley data satisfying the IV inequalities is $\left[1-3.69\times 10^{-6},1\right]$. As mentioned in Section~\ref{subsec:graph_iv} satisfying the IV inequalities implies that fairness is undecidable. In Section~\ref{app:bayes}, we carry out a sensitivity analysis by varying the chosen prior. We also report results on a different dataset from \citet{Bol23} that investigates sex-based discrimination in awarding cum-laude distinctions to graduate students.

In addition to the Bayesian test, the maximum likelihood (ML) estimator satisfies the IV inequalities, implying that there isn't enough evidence to reject the null hypothesis when doing a likelihood ratio test. An implementation of \citet{WangRR17} for the Berkeley dataset also does not reject the null hypothesis (see Section~\ref{app:bayes} for details). The code can be found at {\small\texttt{https://github.com/SourbhBh/BerkeleyCode}}.


\section{Discussion}\label{sec:discussion}
The Berkeley admissions case is a canonical example in the causal fairness literature. \citet{BickelHO75} reached the conclusion of there being no evidence to reject fairness, although under the unrealistic assumption of no unobserved confounding. When allowing for unobserved confounding, we arrived at a different conclusion: since there is very strong evidence that the data satisfies the IV inequalities, it is undecidable from the available data whether the admission procedure was fair or discriminated against sex.

While our analysis was centered around the Berkeley case, there are multiple aspects that generalize---a) \ifdefined \SINGLE The family of causal models we consider \else $\modelsedgerelax$ \fi can be thought of as a mediator with a confounder between mediator and outcome, which is common in mediation analysis. b) The approach of fairness notions being causal hypotheses, with respect to the class of models defined by modeling assumptions, that need to be translated into statistical tests to be useful in practice. c) The observation that for the case of inequality constraints on observational data, a straightforward Bayesian testing procedure is available.

\textbf{Generalization: Non-binary variables}: While the Berkeley dataset had a binary protected attribute and a binary decision outcome, our approach can be generalized to non-binary variables. The response-function parametrization of $\modeliv$ can be used to characterize the set of induced observed distributions of $X,Y$ and $Z$ as a convex polyhedral set. Using computer algebra, this polyhedral set can be characterized by linear inequality constraints even when the alphabet sizes are arbitrary. However, the number of linear inequality constraints quickly explodes for non-binary $Z$ (e.g. for binary $X$ and $Y$, and for $|\cZ|=2,3,4$
 and $5$, we get $12,48,160$ and $420$ 
 inequalities). Nevertheless, these sets of inequality constraints can be tested using a Bayesian testing procedure akin to the one we outline in Section~\ref{sec:bayesiantest},
 thus giving a statistical test for all our fairness notions since the statement of Theorem~\ref{thm:equivalence} holds for any alphabet size.


\textbf{Generalization: Relaxing unobserved confounding assumptions} In Section~\ref{sec:confounder} we made the assumption of no unobserved confounding between $S$ and any other variable. In Section~\ref{app:SDconf} we show that our main theorems, Theorem~\ref{thm:iv_tight} and Theorem~\ref{thm:equivalence} hold even in the case of allowing for confounding between $\sex$ and $\dept$. If unobserved confounding is allowed between $\sex$ and $\outcome$, the set of observational distributions induced by causal models where the direct effect between sex and admissions outcome is absent is the entire simplex which is the same as the set of observational distributions induced by causal models where the direct effect between sex and admissions outcome is present. Therefore, from observational data, it is not possible to distinguish between the presence/absence of a direct effect from sex to admissions outcome if we allow for confounding between sex and admissions outcome.

\textbf{Unmeasured Mediators}: Our conclusions and fairness notions are with respect to the measured variables. The presence of unmeasured mediators could change the interpretation of the results. For example, an unmeasured mediator that is not a `protected' variable, such as choice of undergraduate department, would result in a direct effect of sex on admissions outcome in the marginalized causal model, but might still be considered `fair'. Therefore, even if the data does not satisfy the IV inequalities, we can only claim the existence of the direct effect from sex to admissions outcome and not whether there is unfairness. However, if the existence of latent unprotected mediators is ruled out, then we can consider the existence of a direct effect as `unfair'. 

\textbf{Undecidability}:
Our analysis can be viewed as a Bayesian model comparison between causal models (i) without a direct effect of $S$ on $A$, vs.\ (ii) with a direct effect of $S$ on $A$. If we only have access to observational data then causal models in (ii) result in a saturated model in the observed distribution space, i.e., the set of induced observational distributions is the entire simplex, whereas for the causal models in (i), the IV inequalities hold. Therefore, satisfying the IV inequalities implies that fairness is `undecidable' since causal models in both (i) and (ii) satisfy the IV inequalities. Violating the IV inequalities, however, would imply that there is a direct effect of sex on admissions outcome. Only with the additional modeling assumption of there being no unprotected mediators between $S$ and $A$, could it then be concluded that the admissions process was `unfair'.

\textbf{Selection Bias: } \texttt{UCBAdmissions} dataset only has data from the $6$ largest departments as opposed to $85$ in \cite{BickelHO75}. Also, the fraction of female students is significantly smaller than the fraction of male students. Hence, it is plausible that (latent) selection mechanisms alter the causal model resulting in violating the assumptions of, for instance, absence of a bidirected edge in the causal graph between $\sex$ and $\outcome$ \citep{ChenZM24}. Since allowing for selection bias enlarges the model class $\modelsedge$, and given that the data satisfies the IV inequalities, we conclude that allowing for selection bias will not change our conclusion. 

\begin{acknowledgements}
This work was supported by Booking.com. We thank an anonymous reviewer for their feedback which helped us modulate some of our conclusions.
\end{acknowledgements}

\clearpage

\bibliography{biblio.bib}

\newpage

\onecolumn

\title{Revisiting the Berkeley Admissions data: Statistical Tests for Causal Hypotheses\\(Supplementary Material)}
\maketitle
\appendix

\section{Additional Preliminaries}
\begin{definition}[Twin SCM]
    Let $\model = \Paren{V,W,\cX,P,f}$ be an SCM. The twinning operation maps $\model$ to the \textbf{twin SCM}
    \begin{equation*}
        \model^{\text{twin}} \triangleq \Paren{V \cup V',W,\cX_V \times \cX_{V'} \times \cX_W, P,\tilde{f}}
    \end{equation*}
    where $V' = \left \lbrace v': v \in V \right \rbrace$ is a disjoint copy of $V$ and the causal mechanism $\tilde{f}:\cX_V \times \cX_{V'} \times \cX_W \mapsto \cX_V \times \cX_{V'}$ is given by $\tilde{f}\Paren{x_V,x_{V'},x_W} = \Paren{f(x_V,x_W),f(x_{V'},x_W)}$.
\end{definition}

\begin{definition}[Solution function]
Let $\model = \Paren{\enop, \exrv, \spc, f,P}$ be an acyclic SCM and $C \subseteq \enop$. A \textbf{solution function} of $\model$ with respect to $C$ is a measurable mapping $g_C : \cX_{\enop \backslash C} \times \spc_{\exrv} \mapsto \spc_{C}$ that satisfies the structural equations for $C$, i.e., for all $x_{\enop \backslash C} \in \spc_{\enop \backslash C}$, $P(X_{\exrv})$-a.a $x_{\exrv} \in \spc_{\exrv}$, 
\begin{equation*}
    g_{C}\Paren{x_{\enop\backslash C},x_{\exrv}} = f_{C}\Paren{x_{\enop \backslash C},g_{C}\Paren{x_{\enop\backslash C},x_{\exrv}},x_{\exrv}}.
\end{equation*}
\end{definition}

\begin{definition}[Markov kernels]
    Let $\mathcal{T}$ and $\mathcal{W}$ be measurable spaces. A Markov kernel is defined as a measurable map $K : \mathcal{T} \mapsto \mathcal{P}\Paren{\mathcal{W}}$ where $\mathcal{P}\Paren{\mathcal{W}}$ is defined as the space of probability measures on $\mathcal{W}$.
\end{definition} 

\section{IV Inequalities Expressed as Markov Kernels}\label{app:iv}
For \eqref{eq:iv} to be well defined, we required that $P_{\model}(Z=z) > 0$ for all $z$ for any $\model \in \modeliv$ (see Definition \eqref{def:IVmodelclass}). In this section, we relax this requirement by noting that, in fact, IV inequalities are more appropriately expressed in terms of $P_{\model}\Paren{X,Y \mid \doop{Z}}$. 
\begin{lemma}\label{lem:iv_mk}
Let $\modelivrelax \triangleq \left\lbrace \model: G(\model) \text{ is a subgraph of Figure }\ref{fig:iv} \right \rbrace$. For any $\model \in \modelivrelax$, 
\begin{equation}\label{eq:iv_MK}
     \max_{x} \sum_{y} \max_{z} P_{\model}\Paren{X=x,Y=y\mid \doop{Z=z}} \leq 1. 
\end{equation}
\end{lemma}
\begin{proof}
    Since 
    \begin{align*}
        P_{\model}(X=x,Y=y \mid \doop{Z=z}) &= P_{\model}(f_X(z,U)=x,f_Y(x,U)=y) \\
        &\le P(f_Y(x,U)=y) = P_{\model}(Y=y \mid \doop{X=x}), 
    \end{align*}
\begin{equation}
    \max_{x} \sum_{y} \max_{z} P_{\model}\Paren{X=x,Y=y\mid \doop{Z=z}} \leq \max_{x} \sum_{y} P_{\model}(Y=y \mid \doop{X=x}) = 1.
\end{equation}
\end{proof}
Note that \eqref{eq:iv_MK} is defined even when $ \exists z \in \cZ$ such that $P(Z=z) = 0$. In contrast, positivity must be assumed in \eqref{eq:iv} for the terms to be well-defined. Further, if positivity is assumed, then $\modelivrelax = \modeliv$ and \eqref{eq:iv_MK}  is identical to \eqref{eq:iv}.

\section{Proofs for Section 3}
\subsection{Nested Fairness Notions: Without Confounding}\label{app:nested-nocf}
\unconfnested*
\begin{proof}
\bm{$\nullgraphunconf = \nullctrfunconf$}: We first show that $\nullgraphunconf \subseteq \nullctrfunconf$. $\model \in \nullgraphunconf$ implies $\forall \ldept$, $f_A(s,d,U_A)$ is constant in $\lsex$ $P$-a.s. Therefore, for all $\ldept, \lsex$, \begin{equation}\label{eq:nocf-ctrf}
P_{\model}\Paren{f_A(s,d,U_A)=f_A(S,d,U_A)}=1.
\end{equation}
Therefore, $\model \in \nullctrfunconf$. For the converse, $\model \in \nullctrfunconf$ implies \eqref{eq:nocf-ctrf}. For $s\neq s'$, and all $d$,
\begin{equation*}
P_{\model}\Paren{f_A(s,d,U_A)=f_A(S,d,U_A)}=P_{\model}\Paren{f_A(s,d,U_A)=f_A(s',d,U_A)}P_{\model}\Paren{S=s'} +P_{\model}\Paren{S=s}.
\end{equation*}
From \eqref{eq:nocf-ctrf} if $P_{\model}\Paren{S=s'} > 0$, we conclude $P_{\model}\Paren{f_A(s,d,U_A)=f_A(s',d,U_A)} = 1$. If $P_{\model}\Paren{S=s'} = 0$, since \eqref{eq:nocf-ctrf} holds for $s'$, i.e., for all $d,s'$, $P_{\model}\Paren{f_A(s',d,U_A)=f_A(S,d,U_A)}=1$, we have $$P_{\model}\Paren{f_A(s',d,U_A)=f_A(s,d,U_A)} = 1.$$ Therefore, $\model \in \nullgraphunconf$.

\bm{$\nullctrfunconf \subset \nullinterunconf$}: For $\model \in \nullctrfunconf$, \eqref{eq:nocf-ctrf} implies $P_{\model}\Paren{f_A(s,d,U_A)} = P_{\model}\Paren{f_A(S,d,U_A)}$ for all $d,s$, implying $\model \in \nullinterunconf$. Since Example~\ref{ex:unconfexample} belongs to  $\nullinterunconf \backslash \nullctrfunconf$, the inclusion is strict.

\bm{$\nullinterunconf \subset \nullobsunconf$}: For $\model \in \nullinterunconf$, $P_{\model}\Paren{A=1 \mid \doop{D=d},\doop{S=s}}$ is constant in $s$. Consider a pair $\lsex,\ldept,$ such that, for $\model \in \nullinterunconf$, $P_{\model}\Paren{\lsex,\ldept}>0$. Then 
\begin{equation}\label{eq:inter-obs-nocf}
P_{\model}\Paren{A=1 \mid \doop{D=d},\doop{S=s}} = P_{\model}\Paren{A=1 \mid D=d,S=s}.
\end{equation}
Note that, if, for $\lsex' \neq \lsex$, $P_{\model}\Paren{\lsex',\ldept}=0$, then $P_{\model}\Paren{A=1 \mid D=d,S=s} = P_{\model}\Paren{A=1 \mid D=d}$. If instead, $P_{\model}\Paren{\lsex',\ldept}>0$, then from \eqref{eq:inter-obs-nocf} for $S=\lsex'$, we have that $$P_{\model}\Paren{A=1 \mid D=d,S=s} =P_{\model}\Paren{A=1 \mid D=d,S=s'} = P_{\model}\Paren{A=1 \mid D=d}.$$ Therefore, we conclude that $\model \in \nullobsunconf$ implying $\nullinterunconf \subseteq \nullobsunconf$. Since the SCM in Example~\ref{ex:posexample} lies in $\nullobsunconf \backslash \nullinterunconf$, $\nullinterunconf \subset \nullobsunconf$.

If for all $s,d$, $P_{\model}\Paren{s,d} > 0$, then for $\model \in \nullobsunconf$, $$P_{\model}\Paren{A=1 \mid D=d,S=s} =P_{\model}\Paren{A=1 \mid \doop{D=d},\doop{S=s}}$$ is constant in $s$ and equal to $P_{\model}\Paren{A=1 \mid \doop{D=d}}$. This implies $\model \in \nullinterunconf$. Therefore, if for all $s,d$, $P_{\model}\Paren{s,d} > 0$, then $\nullinterunconf = \nullobsunconf$.


\end{proof}

\subsection{Equivalence of Tests Without Confounding}\label{app:equiv-nocf}

\unconfequiv*

\begin{proof}
    From Lemma~\ref{lem:notion_equiv}, $\distgraphunconf = \distctrfunconf \subseteq \distinterunconf \subseteq \distobsunconf$. Therefore, it suffices to prove that $\distgraphunconf = \distobsunconf$. For every $P_{\model} \in \distobsunconf$, $$P_{\model}\Paren{\outcome, \sex,\dept} = P_{\model}\Paren{\sex} \otimes P_{\model}\Paren{\dept \mid \sex} \otimes P_{\model}\Paren{\outcome \mid \dept}.$$ Hence, $\exists \tilde{\model} \in \nullgraphunconf$ such that $P_{\model}\Paren{A,D,S} = P_{\tilde{\model}}\Paren{A,D,S}$. 
\end{proof}

\section{Proofs for Section 4}
\subsection{Sharpness of IV inequalities}\label{app:ivsharp}
\begin{figure}[th]
     \centering
            \begin{tikzpicture}
            \tikzstyle{vertex}=[circle,fill=none,draw=black,minimum size=17pt,inner sep=0pt]
\node[vertex] (Z) at (0,0) {$Z$};
\node[vertex][fill=lightgray] (U) at (0,1) {$U_Z$};
\node[vertex] (Y) at (3,0) {$Y$};
\node[vertex] (X) at (1.5,0) {$X$};
\node[vertex][fill=lightgray] (R) at (2.2,1) {$R$};
\path (Z) edge (X);
\path (X) edge (Y);
\path (U) edge (Z);
\path (R) edge (X);
\path (R) edge (Y);

            \end{tikzpicture}
        \caption{Response-function parameterization of $M \in \modeliv$} 
        \label{fig:proof_iv}
        \end{figure}

\ivtight*

\begin{proof}We prove a more general statement that includes Theorem~\ref{thm:iv_tight} as a special case.

\begin{lemma}\label{lem:mk_iv_tight}
 Let $X,Y,Z$ be discrete random variables defined on $\cX,\cY,\cZ$ respectively, with $|\cX| = n\geq 2, |\cY|=2, |\cZ| =2$. Define $\mkiv \triangleq \left \lbrace K(X,Y \mid Z) : K(X,Y \mid Z)\text{ satisfies }\eqref{eq:iv_MK} \right \rbrace$. Define $\mkmodeliv \triangleq \left \lbrace P_{\model}\Paren{X,Y \mid \doop{Z}} : \model \in \modelivrelax \right \rbrace.$ Then $\mkiv = \mkmodeliv.$   
\end{lemma}

Note that $\distiv = \left \lbrace P(Z): \forall z, P(Z=z)>0  \right \rbrace \otimes \mkiv$ since assuming positivity, \eqref{eq:iv} is identical to \eqref{eq:iv_MK}. Further, $\distmodeliv = \left\{ P_{\model}(Z): \model \in \modeliv \right\} \otimes \mkmodeliv$ since assuming positivity, $\modelivrelax = \modeliv$ and $P_{\model}\Paren{X,Y \mid \doop{Z}} = P_{\model}\Paren{X,Y \mid Z}$ for $\model \in \modeliv$. Since the first factors are identical, Theorem~\ref{thm:iv_tight} follows from Lemma~\ref{lem:mk_iv_tight}.
\end{proof}

\begin{proof}[Proof of Lemma~\ref{lem:mk_iv_tight}]

For $\model \in \modelivrelax$, the response-function parameterization yields a counterfactually equivalent SCM \cite[Section 8.4]{ForreMooij25} $\tilde{\model} = (\enop,\tilde{\exrv},\tilde{\spc},\tilde{f},\tilde{P})$, where $\enop = \left \lbrace Z,X,Y \right \rbrace, \tilde{\exrv} = \left \lbrace \response, U_Z\right \rbrace, \tilde{\spc} = \spc_{\enop}\times\spc_{\tilde{\exrv}}, \tilde{f} = \Paren{\tilde{f}_{Z}, \tilde{f}_{X}, \tilde{f}_{Y}}$ where we define $\spc_{\response}, \tilde{f},\tilde{P}$ through the function $\Phi: \spc_{\exrv} \mapsto \spc_{\tilde{\exrv}}$ where
\begin{align*}
    \spc_{\response} &\triangleq \cX^{\cZ} \times \cY^{\cX},\\
    \forall u_Z,u_X,u_Y,u, \Phi\Paren{u_Z,u_X,u_Y,u} &\triangleq \Paren{\Paren{z \mapsto f_X(z,u,u_X),x \mapsto f_Y(x,u,u_Y)},u_Z},\\
    \forall u_Z, \tilde{f}_{Z}(u_Z) &\triangleq f_{Z}(u_Z),\\
    \forall \respfunc, z, \tilde{f}_{X}\Paren{\respfunc,z} &\triangleq \respfunc_1\Paren{z}, \\
    \forall \respfunc, x, \tilde{f}_{Y}\Paren{\respfunc,x} &\triangleq \respfunc_2\Paren{x},
    \end{align*}
where $\respfunc = \Paren{\respfunc_1,\respfunc_2}$ and $\tilde{P}$ is the push-forward distribution $\Phi_{*}(P)$.
Note that $\spc_{\response}$ is a discrete space, $\response$ a discrete random variable, and $\tilde{P}(\response)$  a discrete distribution over $\spc_{\response}$.

Under the response-function parameterization, only $\tilde{P}(R)$ is a parameter. We will consider $\tilde{P}(R)$ to be an element of $\mathbb{R}^{n_X^{n_Z} n_Y^{n_X}}$
where $\#\cX = n_X, \#\cZ = n_Z,\#\cY = n_Y$ and  $$\cK_{\tilde{\mathbb{M}}_{\text{IV}+}} \triangleq \left \lbrace P_{\tilde{\model}}(X,Y \mid \doop{Z}): \model \in \modelivrelax \right \rbrace$$ to be a subset of $\mathbb{R}^{n_X n_Y n_Z}$. Note that because of the counterfactual equivalence of the response-function parameterization, which in turn implies interventional equivalence, $\cK_{\tilde{\mathbb{M}}_{\text{IV,r}}} = \mkmodeliv$. From Lemma~\ref{lem:iv_mk}, $\mkmodeliv \subseteq \mkiv$. 

To show the converse, we show that each extreme point of $\mkiv$ is obtained by a point in $\mkmodeliv$. We enumerate all extreme points of $\mkiv$ in Lemma~\ref{lem:extremepointsiv}. We show that each such extreme point is obtained by the following response-function. Choose $x, x' \in \cX, y,y' \in \cY$ with $y=y'$ if $x=x'$. Then any response function satisfying 
\begin{align*}
    r_1(z) &= \begin{cases} 
    x \quad z=0 \\
x' \quad z=1
    \end{cases}\\
      r_2(\tilde{x}) &= \begin{cases} 
    y \quad \tilde{x}=x \\
y' \quad \tilde{x}=x'\\
\text{arbitrary} \quad \text{otherwise}
    \end{cases}
\end{align*}
gives all extreme points of $\mkiv$. Therefore, $\mkmodeliv \supseteq \mkiv$, implying $\mkmodeliv = \mkiv$. 

\end{proof}

\begin{lemma}\label{lem:extremepointsiv}
    Consider the real vector space $\mathbb{R}^{n_X n_Y n_Z}$ spanned by the canonical basis vectors $$\left \lbrace \delta_{x,y|z}: x \in \cX, y \in \cY, z \in \cZ\right \rbrace$$
    where $\delta_{x,y|z}$ denotes a unit vector of length $n_Xn_Yn_Z$ where all entries except the one at $(x,y,z)$ are zero. For $n_Y=n_Z=2, n_X = n\geq 2$, $\mkiv$ considered as a subset of this vector space is a polyhedral set with extreme points 
\begin{equation*}
   \EE = \left \lbrace \delta_{x,y|0} + \delta_{x',y'|1}: x,x' \in \cX; y,y' \in \cY : x\neq x' \right \rbrace \cup \left \lbrace \delta_{x,y|0} + \delta_{x,y|1}: x \in \cX, y\in \cY \right \rbrace.
\end{equation*}
\end{lemma}
\begin{proof}
Consider $\mkiv$ to be a subset of $\mathbb{R}^{n_X n_Y n_Z}$ where each element of $\mkiv$ is represented as $\left \{ K\Paren{X=x, Y=y \mid Z=z} \right \}_{x \in \cX, y \in \cY, z \in \cZ}$ and satisfies
\begin{align*}
\forall x \in \cX : K\Paren{X=x,Y=0|Z=0} + K\Paren{X=x,Y=1|Z=1} &\leq 1,\\
\forall x \in \cX : K\Paren{X=x,Y=0|Z=1} + K\Paren{X=x,Y=1|Z=0} &\leq 1,\\
\forall x \in \cX, \forall y \in \cY, \forall z \in \cZ : K\Paren{X=x,Y=y|Z=z} &\geq 0,\\
\forall z \in \cZ : \sum_{x,y}K\Paren{X=x,Y=y|Z=z} &=1.
\end{align*}
  An elementary result in convex geometry states that a point is an extreme point of a polyhedral set (defined by a set of linear (in)equality constraints) in $\mathbb{R}^m$ vector space if and only if it is a feasible point and there exists a subset $A$ of $m$ constraints that are active and linearly independent.
For the case at hand, for a point to be an extreme point of $\mkiv$, it has to satisfy the above $6n+2$ constraints (feasibility) and additionally, a subset of $n_Xn_Yn_Z = 4n$ of them must be active and linearly independent.
The two normalization constraints are equality constraints and hence must be active at any feasible point.

We first show that all points in $\EE$ are extreme points of $\mkiv$.
First, choose $x,x' \in \cX$ with $x \neq x'$ and choose $y,y' \in \cY$.
It can be verified that $\delta_{x,y|0} + \delta_{x',y'|1}$ satisfies the IV inequalities with $2$ out of the $2n$ IV inequalities being active.
Further, $4n-2$ non-negativity constraints and both normalization constraints are active.
For the subset $A$ of active constraints we can take e.g.\ both active IV inequalities together with the $4n_X - 2$ active nonnegativity constraints; one can check that these are linearly independent.
Second, choose $x \in \cX, y \in \cY$. 
It can be verified that $\delta_{x,y|0} + \delta_{x,y|1}$ satisfies the IV inequalities with $2$ out of the $2n$ IV inequalities being active.
Further, $4n-2$ non-negativity constraints and both normalization constraints are active.
For the subset $A$ of active constraints we can take e.g.\ both active IV inequalities together with the $4n_X - 2$ active nonnegativity constraints; one can check that these are linearly independent.

Finally, we check whether $\EE$ exhausts all extreme points.
Pick a feasible point $b \in \mathbb{R}^{n_X n_Y n_Z}$.
We will refer to the indices $i$ with $b_i \neq 0$ ($b_i = 0$) as the ``non-zero (zero) entries of $b$'', and to the indices $j$ in active constraint $a^T b = 1$ with $a_j \neq 0$ as the ``active entries of the constraint''. 
For a set of entries of $b$ and a set of active entries of one or more active constraints, we refer to their intersection as their ``overlap''.
We also call the $b_i$ corresponding to $K (X, Y | Z = z)$ a ``stratum''.

Suppose $b$ is an extreme point of $\mkiv$.
Then at least $4n - 2$ of the $6n$ inequality constraints must be active.
Hence if exactly $r$ IV inequalities are active at $b$, then we need at least $4n - 2 - r$ active nonnegativity constraints at $b$, or in other words, $b$ must have at least $4n - 2 - r$ zero entries.
Because both strata of $b$ need to be normalized, $b$ can have at most $4n - 2$ zero entries.
The number of zero entries in $b$ that can overlap with the active entries of the active IV inequalities is upper bounded by $r$; indeed, otherwise there would be an active IV inequality for which both its active entries are zero entries of $b$, contradicting that this IV inequality is active.

We will show that the only possibilities for $b$ are the extreme points that we already identified, proceeding case by case.
  \begin{enumerate}
    \item $r > 2$.
      This yields a contradiction with the normalization constraints.
      Indeed, each entry of $b$ appears in exactly one IV inequality, and each IV inequality contains exactly two active entries (one for each stratum).
      Hence, the sum of those entries of $b$ that correspond with active entries of active IV constraints must be exactly $r$.
      However, by the normalization constraints, the sum over \emph{all} entries of $b$ must be 2.
      This implies that $r \le 2$.
    \item $r = 2$.
      Then $b$ must contain at least $4n-4$ zero entries, of which at most two can overlap with the active entries of the active IV inequalities.
      \begin{enumerate}
        \item If there are no such overlaps, then all other entries of $b$ must zero entries.
          This means that we need the $4n-4$ nonnegativity constraints corresponding to those other entries in the subset $A$, as well as the two normalization and inequality constraints; no other active constraints exist that could be added to $A$.
          But then the constraints in $A$ would not be linearly independent.
          Indeed, the following linear dependence between the active constraints is obtained: adding the coefficient vectors of the two active IV inequalities and those of all active non-negativity constraints yields the same result as adding the coefficient vectors of the two normalization constraints.
          So we arrive at a contradiction.
        \item 
          If there is at least one overlap, 
          then $b_j=1$ with $j$ the other active entry in the active IV inequality with the overlap.
          This implies $2n-1$ zeroes in the stratum of $j$.
          Consider now the other active IV inequality. 
          Since $b_j=1$, the active entry corresponding to that stratum must be a zero entry of $b$. 
          Hence, $b_k=1$ for $k$ the other active entry in this IV inequality.
          This means that $b$ must be of the form $\delta_{x_1,x_2|0} + \delta_{x_1',x_2'|1}$, with the two non-zero entries corresponding to active entries of two different IV inequalities, and hence we recover the extreme points already identified.
      \end{enumerate}
    \item $r = 1$.
      Then $b$ must contain at least $4n-3$ zero entries, of which at most one can overlap with  the active entries of active IV inequalities.

      Because of the normalization constraints, we need at least one non-zero entry in both strata.
      This means that there must be a stratum $x_3$ with exactly one non-zero entry, and then $b$ must be of the form
      $b = \delta_{x_1,x_2|x_3} + \gamma \delta_{x_1',x_2'|x_3'} + (1-\gamma) \delta_{x_1'',x_2''|x_3'}$ with $\gamma \in (0,1]$ and $x_3' \ne x_3$.
      This means that the active IV inequality must be the one that has active entry $(x_1,x_2|x_3)$.
      Its other active entry is then $(x_1,1-x_2|x_3')$, and this must be an overlapping zero entry of $b$.

      If $\gamma = 1$ then this would also activate the IV inequality that contains active entry $(x_1',x_2'|x_3')$.
      Since the only active IV inequality must also contain active entry $(x_1,x_2|x_3)$, this gives a contradiction, as the two coefficients of $b$ at the active entries sum to 2 instead of 1.
      Hence, $\gamma \in (0,1)$, and $b$ contains exactly $4n-3$ zero entries.

      So the active constraints consist of 1 active IV inequality, 2 active normalization constraints and $4n-3$ active nonnegativity constraints.
      However, these are not linearly independent.
      Indeed: subtracting the coefficient vectors of the $2n-1$ nonnegativity constraints in stratum $x_3$ from the coefficient vector of the normalization constraint of that stratum, and adding the coefficient vector of the nonnegativity constraint corresponding to $(x_1,1-x_2|x_3')$ gives a vector that is identical to the coefficient vector of the active IV inequality.

      So we have arrived at a contradiction.
    \item $r = 0$.
      Then $b$ must contain at least $4n-2$ zero entries.
      Because of the normalization constraints, $b$ needs one non-zero entry in both strata, and its value must be 1.
      This will activate at least one IV inequality. 
      Contradiction. \qedhere
  \end{enumerate}
\end{proof}

\subsection{Nested Fairness Notions: With Confounding}\label{app:nested-cf}
\begin{proposition}\label{prop:cfnotions}
    \begin{equation*}
    \nullgraphrelax = \nullctrfrelax \subset \nullinterrelax, 
    \end{equation*}
     \begin{equation*}
    \nullgraph = \nullctrf \subset \nullinter. 
    \end{equation*}
\end{proposition}

\begin{proof}
If $\model \in \modelsedgerelax$, then 
\begin{equation}
    A^{\doop{S=s,D=d}} = f_{\outcome}\Paren{s,d,U_A,U}, \hspace{5mm}A^{\doop{D=d}} = f_{\outcome}\Paren{S,d,U_A,U}.
\end{equation}

For $\model \in \nullgraphrelax$, since $\sex$ is not a parent of $\outcome$, for all $\ldept, f_{A}(s,d,U_A,U)$ is constant in $s$ $P$-a.s. This implies that for all $d,s,s'$, 
\begin{equation*}
P_{\model}\Paren{f_{\outcome}\Paren{s,d,U_A,U} = f_{\outcome}\Paren{s',d,U_A,U}} = 1
\end{equation*}
Therefore, for all $s,d$,
\begin{equation}\label{eq:ctrf_causalmechanism}
P_{\model}\Paren{f_{\outcome}\Paren{s,d,U_A,U} = f_{\outcome}\Paren{S,d,U_A,U}} = 1,
\end{equation}
implying that $\model \in \nullctrfrelax$. For the converse, $\model \in \nullctrfrelax$ implies \eqref{eq:ctrf_causalmechanism}. For $s\neq s'$, and all $d$,
\begin{align*}
&P_{\model}\Paren{f_A(s,d,U_A,U)=f_A(S,d,U_A,U)}\\
&=P_{\model}\Paren{f_A(s,d,U_A,U)=f_A(s',d,U_A,U)}P_{\model}\Paren{S=s'} +P_{\model}\Paren{S=s}.
\end{align*}
If $P_{\model}\Paren{S=s'} > 0$, we conclude $P_{\model}\Paren{f_A(s,d,U_A,U)=f_A(s',d,U_A,U)} = 1$. If $P_{\model}\Paren{S=s'} = 0$, since \eqref{eq:ctrf_causalmechanism} holds for $s'$, we have $P_{\model}\Paren{f_A(s,d,U_A,U)=f_A(s',d,U_A,U)} = 1$. Therefore, $\model \in \nullgraphrelax$. Therefore, $\nullgraphrelax = \nullctrfrelax$. Further, $\nullgraph = \nullctrf$.

We now prove that $\nullctrfrelax \subseteq \nullinterrelax$.
For $\model \in \nullctrfrelax$, \eqref{eq:ctrf_causalmechanism} holds. Therefore, for all $\lsex,\ldept$,
\begin{equation}
P_{\model}\Paren{f_{\outcome}\Paren{s,d,U_A,U}} = P_{\model}\Paren{f_{\outcome}\Paren{S,d,U_A,U}}.
\end{equation}   
Therefore, $\nullctrfrelax  \subseteq \nullinterrelax$ and subsequently $\nullctrf  \subseteq \nullinter$. Note that the Example~\ref{ex:unconfexample} lies in $\modelsedgerelax$ (and in $\modelsedge$ since $P_{\model}(S=s)>0$ for all $s$) for any $U$ that is independent of $U_S,U_D,U_A$. Further, $P_{\model}\Paren{A=1 \mid \doop{S=s},\doop{D=d}} = 0.5 = P_{\model}\Paren{A=1 \mid \doop{D=d}}$; however, $S$ is a parent of $A$. Therefore, $\nullctrfrelax \subset \nullinterrelax$ and $\nullctrf \subset \nullinter$.



\end{proof}

\subsection{Equivalence of Statistical Tests}\label{app:equiv-cf}

 \begin{figure}[th]
     \centering
            \begin{tikzpicture}
            \tikzstyle{vertex}=[circle,fill=none,draw=black,minimum size=17pt,inner sep=0pt]
\node[vertex] (Z) at (0,0) {$S$};
\node[vertex][fill=lightgray] (U) at (0,1) {$U_S$};
\node[vertex] (Y) at (2.5,0) {$A$};
\node[vertex] (X) at (1.3,1) {$D$};
\node[vertex][fill=lightgray] (R) at (2.4,0.95) {$R$};
\path (Z) edge (X);
\path (X) edge (Y);
\path (U) edge (Z);
\path (R) edge (X);
\path (R) edge (Y);
\path (Z) edge (Y);

            \end{tikzpicture}
        \caption{Response-function parameterization of $M \in \modelsedgerelax$} 
        \label{fig:proof_conf}
        \end{figure}

\confequiv*

\begin{proof}Like in Section~\ref{app:ivsharp}, we prove a more general statement, Lemma~\ref{lem:mk_equivalence}, that includes Theorem~\ref{thm:equivalence} as a special case. 
We define analogues of $\modelsedge, \nullnotion$ that remove the positivity assumption, $P_{\model}(S=s)>0$ for all $s$, as $\modelsedgerelax, \nullnotionrelax$, respectively (where we use `notion' as a placeholder for `graph', `ctrf' and `inter').
\begin{lemma}\label{lem:mk_equivalence}Let 
\begin{align*}
\mkgraph &\triangleq \left \lbrace P_{\model}\Paren{\dept,\outcome \mid \doop{\sex}} : \model \in \nullgraphrelax \right \rbrace, \\
\mkinter &\triangleq \left \lbrace P_{\model}\Paren{\dept,\outcome \mid \doop{\sex}} : \model \in \nullinterrelax \right \rbrace, \\
\mkctrf &\triangleq \left \lbrace P_{\model}\Paren{\dept,\outcome \mid \doop{\sex}} : \model \in \nullctrfrelax \right \rbrace.
\end{align*}
Then $\mkinter = \mkctrf = \mkgraph = \mkiv,$ where $\mkiv$ is defined in Lemma~\ref{lem:mk_iv_tight}.
\end{lemma}

Note that $\distiv = \left \lbrace P(Z): \forall z, P(Z=z)>0  \right \rbrace \otimes \mkiv$ since assuming positivity, \eqref{eq:iv} is identical to \eqref{eq:iv_MK}. Further, $\distnotion = \left\{ P_{\model}(\sex):\model \in \nullnotion  \right\} \otimes \mknotion$ since assuming positivity, $\modelsedge = \modelsedgerelax$ and $P_{\model}\Paren{\dept,\outcome \mid \doop{\sex}} = P_{\model}\Paren{\dept,\outcome \mid \sex}$ for $\model \in \modelsedge$. Since the first factors are identical, Theorem~\ref{thm:equivalence} follows from Lemma~\ref{lem:mk_equivalence}.
\end{proof}


\begin{proof}[Proof of Lemma~\ref{lem:mk_equivalence}]


For $\model \in \modelsedgerelax$, the response-function parameterization yields a counterfactually equivalent SCM, $\tilde{\model}$ represented by the tuple $(\enop,\tilde{\exrv},\tilde{\spc},\tilde{f},\tilde{P})$, where $\enop = \left \lbrace \sex, \dept, \outcome \right \rbrace, \tilde{\exrv} = \left \lbrace \response, U_{\sex} \right \rbrace, \tilde{\spc} =\spc_{\enop}\times\spc_{\tilde{\exrv}}, \tilde{f} = \Paren{\tilde{f}_{\sex}, \tilde{f}_{\dept}, \tilde{f}_{\outcome}}$ where we define $\spc_{\response}, \tilde{f},\tilde{P}$ through the function $\Phi: \spc_{\exrv} \mapsto \spc_{\tilde{\exrv}}$ where
\begin{align*}
    \spc_{\response} &\triangleq \spc_{\dept}^{\spc_{\sex}} \times \spc_{\outcome}^{\spc_{\sex}\times \spc_{\dept}},\\
    \forall u_S,u_D,u_A,u, \Phi\Paren{u_S,u_D,u_A,u} &\triangleq \Paren{\Paren{s \mapsto f_D(s,u,u_D),(s,d) \mapsto f_A(s,d,u,u_A)},u_S},\\
    \forall u_{\sex}, \tilde{f}_{\sex}\Paren{u_S}&\triangleq f_{\sex}(u_{\sex}),\\
\forall \lsex, \tilde{f}_{\dept}\Paren{\respfunc,\lsex} &\triangleq \respfunc_1\Paren{\lsex}, \\
\forall \lsex, \ldept, \tilde{f}_{\outcome}\Paren{\respfunc,\lsex,\ldept} &\triangleq \respfunc_2\Paren{\lsex,\ldept},
\end{align*}
where $\respfunc = \Paren{\respfunc_1,\respfunc_2}$ and $\tilde{P}$ is the push-forward distribution $\Phi_{*}(P)$.
Note that $\spc_{\response}$ is a discrete space, $\response$ a discrete random variable, and $\tilde{P}(\response)$  a discrete distribution over $\spc_{\response}$. Under the response-function parameterization, only $\tilde{P}(\response)$ is a parameter and we will abuse notation and denote it as $\tilde{P}$ henceforth. Therefore, we can represent $\nullgraphrelax$ in the parameter space as 


\begin{equation}\label{eq:respfunc_graph_edge}
    \nullgraphresp 
    \triangleq \left \lbrace \tilde{P} \in \triangle\Paren{\cX_{\response}} : \tilde{P}\Paren{\respfunc_1,\respfunc_2} \neq 0 \text{ implies } \forall \ldept, \respfunc_2(0,\ldept) = \respfunc_2\Paren{1,\ldept} \right \rbrace.
\end{equation}

To express $\nullinterrelax$, we express the interventional Markov kernels $P_{\tilde{\model}}\Paren{\outcome\mid \doop{\sex}, \doop{\dept}}$ in terms of $\tilde{P}$. Since counterfactual equivalence implies interventional equivalence, for all $\lsex, \ldept$, $P_{\model}\Paren{\outcome=1\mid \doop{\sex=\lsex}, \doop{\dept=\ldept}} = P_{\tilde{\model}}\Paren{\outcome=1\mid \doop{\sex=\lsex}, \doop{\dept=\ldept}}$, where 
\begin{align}
    P_{\tilde{\model}}\Paren{\outcome=1\mid \doop{\sex=\lsex}, \doop{\dept=\ldept}} &= \sum\limits_{\Paren{\respfunc_1,\respfunc_2}  \in \cX_{\response}}\bm{1}\Brack{\respfunc_2\Paren{\lsex,\ldept}=1}\tilde{P}\Paren{\respfunc_1,\respfunc_2}, \label{eq:inter_resp}\\
    P_{\tilde{\model}}\Paren{\outcome=1\mid \doop{\dept=\ldept}} &= \sum\limits_{\lsex^*}\sum\limits_{\Paren{\respfunc_1,\respfunc_2}  \in \cX_{\response}}\bm{1}\Brack{\respfunc_2\Paren{\lsex^*,\ldept}=1}\tilde{P}\Paren{\respfunc_1,\respfunc_2}P_{\tilde{\model}}\Paren{\lsex^*} \label{eq:inter_resp_doD},
\end{align}
Subtracting \eqref{eq:inter_resp} from  \eqref{eq:inter_resp_doD} we get 
\begin{align}
    &P_{\tilde{\model}}\Paren{\outcome=1\mid \doop{\sex=\lsex}, \doop{\dept=\ldept}} - P_{\tilde{\model}}\Paren{\outcome=1\mid \doop{\dept=\ldept}} \nonumber \\
    & =\Paren{\sum\limits_{\Paren{\respfunc_1,\respfunc_2}  \in \cX_{\response}} \Paren{\bm{1}\Brack{\respfunc_2\Paren{0,\ldept}=1} - \bm{1}\Brack{\respfunc_2\Paren{1,\ldept}=1}}\tilde{P}\Paren{\respfunc_1,\respfunc_2}}P_{\tilde{\model}}\Paren{s'} = 0 \label{eq:inter_resp_s}
\end{align}
for $\model \in \nullinterrelax$, where $s' \neq s$. Similarly, 
\begin{align}
    &P_{\tilde{\model}}\Paren{\outcome=1\mid \doop{\sex=\lsex'}, \doop{\dept=\ldept}} - P_{\tilde{\model}}\Paren{\outcome=1\mid \doop{\dept=\ldept}} \nonumber\\
    & =\Paren{\sum\limits_{\Paren{\respfunc_1,\respfunc_2}  \in \cX_{\response}} \Paren{\bm{1}\Brack{\respfunc_2\Paren{0,\ldept}=1} - \bm{1}\Brack{\respfunc_2\Paren{1,\ldept}=1}}\tilde{P}\Paren{\respfunc_1,\respfunc_2}}P_{\tilde{\model}}\Paren{s} = 0 \label{eq:inter_resp_sp}.
\end{align}
Since both \eqref{eq:inter_resp_s} and \eqref{eq:inter_resp_sp} hold, 
 the response-function parameterized analogue of $\nullinterrelax$ is 
\begin{equation}\label{eq:respfun_inter_edge}
    \nullinterresp \triangleq \left \lbrace \tilde{P} \in \triangle\Paren{\cX_{\response}} : \forall \ldept, \sum\limits_{\Paren{\respfunc_1,\respfunc_2}  \in \cX_{\response}} \Paren{\bm{1}\Brack{\respfunc_2\Paren{0,\ldept}=1} - \bm{1}\Brack{\respfunc_2\Paren{1,\ldept}=1}}\tilde{P}\Paren{\respfunc_1,\respfunc_2} = 0 \right \rbrace. 
\end{equation}

Note that both $\nullgraphresp$ and $\nullinterresp$ are polyhedra in $\triangle\Paren{\cX_{\response}}$. Further, $\nullgraphresp \subseteq \nullinterresp$. While, $\nullgraphresp, \nullinterresp$ are collections of distributions, we will also refer to them as collection of response-function-parameterized SCMs. 

From interventional equivalence (which follows as a result of counterfactual equivalence) of the response-function-parameterization, we have 
\begin{align*}
    \mkgraph &= \left \lbrace P_{\tilde{\model}}\Paren{\dept,\outcome\mid \doop{\sex}} : \tilde{\model} \in \nullgraphresp \right \rbrace \\
    \mkinter &= \left \lbrace P_{\tilde{\model}}\Paren{\dept,\outcome\mid \doop{\sex}} : \tilde{\model} \in \nullinterresp \right \rbrace.
\end{align*}


We now show that $\mkinter = \mkgraph = \mkiv$. First, notice that $\mkinter \supseteq \mkgraph$ since $\nullinterresp \supseteq \nullgraphresp$. We first show that $\mkinter \subseteq \mkiv$ and then $\mkgraph = \mkiv$ which concludes the argument. 

\bm{$\mkinter \subseteq \mkiv$}: 
The solution function of the response-function parameterized SCM, $g_{A,D}: \cX_{\sex} \times \cX_{\response} \mapsto \cX_{\outcome} \times \cX_{\dept}$ induces a mapping from $\triangle\Paren{\cX_{\response}}$ which can be considered as a subset of $\RR^{\# \cX_{\response}}$ to the set of Markov kernels $P_{\tilde{\model}}\Paren{\dept,\outcome \mid \doop{\sex}}$ which  can be considered to be a subset of $\RR^{\#\Paren{\cX_{\outcome}}\times \#\Paren{\cX_{\dept}}\times \#\Paren{\cX_{\sex}}}$.
The condition in \eqref{eq:respfun_inter_edge} implies that for all $\ldept$,
\begin{equation}\label{eq:constraint_outcome_one}
    \sum\limits_{\respfunc: \respfunc_2\Paren{0,\ldept}=1} \tilde{P}(\respfunc) = \sum\limits_{\respfunc: \respfunc_2\Paren{1,\ldept}=1} \tilde{P}(\respfunc). 
\end{equation}
Since, $\sum\limits_{\respfunc} \tilde{P}\Paren{\respfunc} = 1$, 
\begin{equation}\label{eq:constraint_outcome_zero}
    \sum\limits_{\respfunc: \respfunc_2\Paren{0,\ldept}=0} \tilde{P}(\respfunc) = \sum\limits_{\respfunc: \respfunc_2\Paren{1,\ldept}=0} \tilde{P}(\respfunc). 
\end{equation}
Denote $P_{\tilde{\model}}\Paren{\dept = \ldept, \outcome = \loutcome \mid \doop{\sex = \lsex}} $ by $P_{\tilde{\model}}\Paren{d,a || s}$. For  $P_{\tilde{\model}}\Paren{d,a || s} \in \mkinter$,
\begin{equation*}
    P_{\tilde{\model}}\Paren{d,a || s} = \sum\limits_{\respfunc: \respfunc_1\Paren{\lsex}=\ldept, \respfunc_2\Paren{\lsex,\ldept} = \loutcome
    } \tilde{P}(\respfunc). 
\end{equation*}
Therefore, from \eqref{eq:constraint_outcome_one}, 
\begin{align}
    \sum\limits_{\respfunc: \respfunc_2\Paren{0,\ldept}=1} \tilde{P}(\respfunc) &= P_{\tilde{\model}}(1,\ldept || 0) + \sum\limits_{\respfunc: \respfunc_1(0) \neq \ldept, \respfunc_2\Paren{0,\ldept}=1} \tilde{P}(\respfunc) \label{eq:1d0}\\
    &= \sum\limits_{\respfunc: \respfunc_2\Paren{1,\ldept}=1} \tilde{P}(\respfunc) \nonumber \\
    &= P_{\tilde{\model}}(1,\ldept || 1) + \sum\limits_{\respfunc: \respfunc_1(1) \neq \ldept, \respfunc_2\Paren{1,\ldept}=1} \tilde{P}(\respfunc) \label{eq:1d1}.
\end{align}
From \eqref{eq:constraint_outcome_zero}, 
\begin{align}
    \sum\limits_{\respfunc: \respfunc_2\Paren{0,\ldept}=0} \tilde{P}(\respfunc) &= P_{\tilde{\model}}(0,\ldept || 0) + \sum\limits_{\respfunc: \respfunc_1(0) \neq \ldept, \respfunc_2\Paren{0,\ldept}=0} \tilde{P}(\respfunc) \label{eq:0d0} \\
    &= \sum\limits_{\respfunc: \respfunc_2\Paren{1,\ldept}=0} \tilde{P}(\respfunc) \nonumber \\
    &= P_{\tilde{\model}}(0,\ldept || 1) + \sum\limits_{\respfunc: \respfunc_1(1) \neq \ldept, \respfunc_2\Paren{1,\ldept}=0} \tilde{P}(\respfunc) \label{eq:0d1}.
\end{align}
Since from \eqref{eq:constraint_outcome_one},
\begin{equation*}
    \sum\limits_{\respfunc} \tilde{P}\Paren{\respfunc} = \sum\limits_{\respfunc: \respfunc_2\Paren{0,\ldept}=0} \tilde{P}(\respfunc) + \sum\limits_{\respfunc: \respfunc_2\Paren{0,\ldept}=1} \tilde{P}(\respfunc) = \sum\limits_{\respfunc: \respfunc_2\Paren{0,\ldept}=0} \tilde{P}(\respfunc) + \sum\limits_{\respfunc: \respfunc_2\Paren{1,\ldept}=1} \tilde{P}(\respfunc) = 1.
\end{equation*}
Substituting from \eqref{eq:0d0} and \eqref{eq:1d1}, 
\begin{equation*}
    P_{\tilde{\model}}(0,\ldept || 0) + \sum\limits_{\respfunc: \respfunc_1(0) \neq \ldept, \respfunc_2\Paren{0,\ldept}=0} \tilde{P}(\respfunc) + P_{\tilde{\model}}(1,\ldept || 1) + \sum\limits_{\respfunc: \respfunc_1(1) \neq \ldept, \respfunc_2\Paren{1,\ldept}=1} \tilde{P}(\respfunc) =1. 
\end{equation*}
Similarly, substituting from \eqref{eq:0d1} and \eqref{eq:1d0}, 
\begin{equation*}
    P_{\tilde{\model}}(0,\ldept || 1) + \sum\limits_{\respfunc: \respfunc_1(1) \neq \ldept, \respfunc_2\Paren{1,\ldept}=0} \tilde{P}(\respfunc)+ P_{\tilde{\model}}(1,\ldept || 0) + \sum\limits_{\respfunc: \respfunc_1(0) \neq \ldept, \respfunc_2\Paren{0,\ldept}=1} \tilde{P}(\respfunc)=1. 
\end{equation*}
 This implies $P_{\tilde{\model}}(0,\ldept || 0) + P_{\tilde{\model}}(1,\ldept || 1) \leq 1, P_{\tilde{\model}}(0,\ldept || 1) + P_{\tilde{\model}}(1,\ldept || 0) \leq 1$. These are precisely the IV inequalities and they are satisfied. Therefore, $\mkinter \subseteq \mkiv$.

 \bm{$\mkgraph = \mkiv$} follows from Theorem~\ref{thm:iv_tight} since $\model \in \nullgraphrelax$ implies $\model \in \modelivrelax$. By Proposition~\ref{prop:cfnotions}, the lemma follows. 

\end{proof}

\subsection{Relaxing the assumption of no confounding between $S$ and $D$}\label{app:SDconf}

In this section, we prove that Theorem~\ref{thm:iv_tight} and Theorem~\ref{thm:equivalence} hold when $\modeliv$ and $\modelsedge$ are expanded by allowing for confounding between $\sex$ and $\dept$. Denote the corresponding expanded models by $\modelivZX$ and $\modelsedgeSD$, respectively, where the former has structural equations of the form $Z = f_Z(U_Z,U_{ZX}), X=f_{X}(Z,U_X,U_{ZX},U_{XY}),Y= f_Y(X,U_Y,U_{XY})$ where $U_Z,U_X,U_Y,U_{ZX},U_{XY}$ are independent exogenous random variables, and the latter has structural equations of the form $S = f_S(U_S,U_{SD}), D=f_{D}(S,U_D,U_{SD},U_{DA}),A= f_A(S,D,U_A,U_{DA})$ where $U_S,U_D,U_A,U_{SD},U_{DA}$ are independent exogenous random variables. The corresponding expanded causal null hypothesis corresponding to the fairness notions are denoted by $H^0_{\text{cf-notion}+SD}$ where we use `notion' as a placeholder for `graph',`ctrf', `inter'. 

For $\model \in \modelivZX$, pick $U \sim \text{Unif}\left[0,1\right]$ and a deterministic map $h: \mathcal{Z} \times \left[0,1\right] \mapsto \mathcal{U}_{ZX}$ such that $U_{ZX} = h(Z,U)$ a.s.. Note that such an $h$ always exists for any random variable $U_{ZX}$ taking values in a standard measurable space (see e.g.\ \cite[Corollary 2.7.7]{ForreMooij25}). Define $\tilde{\model} = (V=\{\tilde{Z},\tilde{X},\tilde{Y} \}, W = \{ U_{\tilde{Z}},U_{\tilde{X}},U,U_{\tilde{Y}},U_{\tilde{XY}}\}, \cX = \cX_{V} \times \cX_W, \tilde{f} = (f_{\tilde{Z}},f_{\tilde{X}},f_{\tilde{Y}}),\tilde{P})$ where a) $\forall u_{\tilde{Z}}, f_{\tilde{Z}}(u_{\tilde{Z}}) \triangleq u_{\tilde{Z}}$, b) $\forall u_{\tilde{X}}, u, u_{\tilde{XY}}, \tilde{z}, f_{\tilde{X}}(u_{\tilde{X}},u,u_{\tilde{XY}},\tilde{z}) \triangleq f_{X}(\tilde{z},h(\tilde{z},u),u_{\tilde{X}},u_{\tilde{XY}})$, c) $\forall u_{\tilde{Y}},u_{\tilde{XY}}, \tilde{x}, f_{\tilde{Y}}(u_{\tilde{Y}},u_{\tilde{XY}}, \tilde{x}) \triangleq f_{Y}(u_{\tilde{Y}},u_{\tilde{XY}}, \tilde{x})$, and $\tilde{P} = P_{\model}(Z) \otimes P_X \otimes \text{Unif}\left[0,1\right] \otimes P_{Y} \otimes P_{XY}$ where $P_X, P_Y, P_{XY}, P_Z, P_{ZX}$ are marginals of $P$ over $U_X, U_Y, U_{XY}, U_Z, U_{ZX}$ respectively.  Note that $P_{\model}(Z)$ is the pushforward of $P_Z \otimes P_{ZX}$ through $f_Z$ thus making $\tilde{P}$ a product distribution over the exogenous random variables. By the above construction, $\tilde{M} \in \modeliv$. For every $ \model \in \modelivZX$, $P_{\model}(X,Y,Z) = P_{\model}(Z) \times P_{\model}(X,Y \mid Z) = P_{\tilde{\model}}(\tilde{Z}) \times P_{\tilde{\model}}(\tilde{X},\tilde{Y} \mid \doop{\tilde{Z}}) \in \distmodeliv$. Therefore, $\{ P_{M}(Z,X,Y): M \in \modelivZX \} = \distmodeliv$. Using a similar argument that replaces the labels $Z,X,Y$ by $S,D,A$, $\{ P_{M}(S,D,A): M \in H^{0}_{\text{cf-notion}+SD}\} = \distnotion$. This implies Theorem~\ref{thm:iv_tight} and Theorem~\ref{thm:equivalence} hold for $\modelivZX$ and $\modelsedgeSD$ respectively.

\section{Comparison With Existing Notions}
\subsection{Without Confounding}
\subsubsection{Counterfactual Fairness and Demographic Parity}\label{app:kusnerctrfdemo}

We restate the counterfactual notion of fairness from \citet{KusnerLRS17} for the Berkeley example below.
\begin{definition}[Counterfactual Fairness \citep{KusnerLRS17}]\label{def:ctrfkusner}
$\model \in \modelsunconfedge$ is fair if for all $\lsex,\ldept$, $P_{\model}\Paren{s,d}>0$ implies
\begin{equation*}
P_{\model}\Paren{\outcome^{\doop{\sex=\lsex'}}\mid \dept = \ldept, \sex= \lsex} = P_{\model}\Paren{\outcome^{\doop{\sex=\lsex}}\mid \dept = \ldept, \sex= \lsex}
\end{equation*}
for $s' \neq s$.
\end{definition}
The counterfactual fairness notion of \cite{KusnerLRS17} implies demographic parity for the Berkeley example without allowing for confounding.
\begin{proposition}\label{prop:ctrf-demo-nocf}
        If $\model \in \modelsunconfedge$ is counterfactually fair according to Definition~\ref{def:ctrfkusner}, then $P_{\model}$ satisfies demographic parity, i.e., for all $s,s'$ such that $P_{\model}(s), P_{\model}(s') >0$, $$P_{\model}\Paren{A=1|S=s} = P_{\model}\Paren{A=1|S=s'}.$$
\end{proposition}
    
\begin{proof}
\begin{figure}[t]
     \centering
            \begin{tikzpicture}
            \tikzstyle{vertex}=[circle,fill=none,draw=black,minimum size=17pt,inner sep=0pt]
\node[vertex] (S) at (0,0) {$S$};
\node[vertex][fill=lightgray] (U_S) at (0,-1) {$U_S$};
\node[vertex] (S') at (0,-2) {$S'$};
\node[vertex] (A) at (3,0) {$A$};
\node[vertex] (A') at (3,-2) {$A'$};
\node[vertex] (D) at (1.5,1.5) {$D$};
\node[vertex] (D') at (1.5,-3.5) {$D'$};
\node[vertex][fill=lightgray] (U_D) at (1.5,-1) {$U_D$};
\node[vertex][fill=lightgray] (U_A) at (3,-1) {$U_A$};
\path (S) edge (D);
\path (D) edge (A);
\path (S) edge (A);
\path (S') edge (A');
\path (S') edge (D');
\path (D') edge (A');
\path (U_S) edge (S);
\path (U_D) edge (D);
\path (U_D) edge (D');
\path (U_A) edge (A);
\path (U_A) edge (A');

            \end{tikzpicture}
        \caption{Causal graph of twin network $(\model^{\text{twin }})^{\doop{S'=s'}}$} 
        \label{fig:twin_net_kusner}
        \end{figure}
    The right-hand side in Definition~\ref{def:ctrfkusner} is $P_{\model}\Paren{A \mid D=d,S=s}$. Therefore, for $s,d$ such that $P_{\model}(s,d) >0$, counterfactual fairness implies that 
    \begin{equation*}
    P_{\model}\Paren{\outcome^{\doop{\sex=\lsex'}}, \dept = \ldept \mid \sex= \lsex} = P_{\model}\Paren{\outcome, \dept = \ldept \mid \sex=\lsex}.
    \end{equation*}
    Marginalizing $\dept$, 
       \begin{equation}\label{eq:marginalDctrf}
    P_{\model}\Paren{\outcome^{\doop{\sex=\lsex'}}\mid \sex= \lsex} = P_{\model}\Paren{\outcome\mid \sex=\lsex}.
    \end{equation}
By Rule $2$ of do-calculus, if $P_{\model}(S=s)>0$, 
\begin{align*}
    P_{\model}\Paren{\outcome\mid \sex=\lsex} &= P_{\model}\Paren{\outcome\mid \doop{\sex=\lsex}}, \\
    &\stackrel{(a)}{=} P_{\model}\Paren{\outcome^{\doop{\sex=\lsex'}}\mid \sex= \lsex}, \\
    &\stackrel{(b)}{=} P_{\model^{\text{twin}}}\Paren{\outcome'\mid \doop{\sex'=\lsex'},\sex= \lsex}, \\
    &\stackrel{(c)}{=} P_{\model}\Paren{\outcome \mid \doop{\sex=\lsex'}}, \\
    &\stackrel{(d)}{=} P_{\model}\Paren{\outcome \mid \sex=\lsex'}.
\end{align*}
    where $(a)$ follows from \eqref{eq:marginalDctrf}, $(b)$ follows from expressing the counterfactual $P_{\model}\Paren{\outcome^{\doop{\sex=\lsex'}}\mid \sex= \lsex}$ in the twin network model, $(c)$ follows from the twin network, and $(d)$ follows from Rule $2$ of do-calculus since $P_{\model}(s')>0$. Therefore, counterfactual fairness implies demographic parity. 
\end{proof}
Note that demographic parity falls prey to Simpson's paradox in the Berkeley example. The above result shows that a valid test for demographic parity is a valid test for \citet{KusnerLRS17}'s counterfactual fairness notion for the assumed model class, $\modelsunconfedge$.
\subsubsection{Path-dependent Counterfactual Fairness}\label{app:kusnerpathnocf}
We next show that testing the path-dependent counterfactual fairness notion given in the appendix of \citet{KusnerLRS17} coincides with a conditional independence test $A \indep S \mid D$. 

\begin{definition}[Path-dependent Counterfactual Fairness \citep{KusnerLRS17}]\label{def:pdctrfkusner}
$\model \in \modelsunconfedge$ is fair if for all $\lsex,\ldept,$ $P_{\model}\Paren{s,d}>0$ implies
\begin{equation*}
P_{\model}\Paren{\outcome^{\doop{\sex = \lsex', \dept = \ldept}}=1 \mid \dept = \ldept, \sex=\lsex} = P_{\model}\Paren{\outcome^{\doop{\sex = \lsex, \dept = \ldept}}=1 \mid \dept = \ldept, \sex=\lsex}
\end{equation*}
for $s' \neq s$.
\end{definition}

\begin{proposition}\label{prop:pathwise}
$$\left \lbrace P_{\model}(A,D,S): \model \in \modelsunconfedge \text{ satisfies path-dependent counterfactual fairness} \right \rbrace = \distobsunconf. $$
\end{proposition}
\begin{proof}
\begin{figure}[th]
     \centering
            \begin{tikzpicture}
            \tikzstyle{vertex}=[circle,fill=none,draw=black,minimum size=17pt,inner sep=0pt]
\node[vertex] (S) at (0,0) {$S$};
\node[vertex][fill=lightgray] (U_S) at (0,-1) {$U_S$};
\node[vertex] (S') at (0,-2) {$S'$};
\node[vertex] (A) at (3,0) {$A$};
\node[vertex] (A') at (3,-2) {$A'$};
\node[vertex] (D) at (1.5,1.5) {$D$};
\node[vertex] (D') at (1.5,-3.5) {$D'$};
\node[vertex][fill=lightgray] (U_D) at (1.5,-1) {$U_D$};
\node[vertex][fill=lightgray] (U_A) at (3,-1) {$U_A$};
\path (S) edge (D);
\path (D) edge (A);
\path (S) edge (A);
\path (S') edge (A');
\path (D') edge (A');
\path (U_S) edge (S);
\path (U_D) edge (D);
\path (U_A) edge (A);
\path (U_A) edge (A');
            \end{tikzpicture}
        \caption{Causal graph of twin network $(\model^{\text{twin }})^{\doop{S'=s',D'=d}}$ for $\model \in \modelsunconfedge$} 
        \label{fig:twin_net_pathwise}
        \end{figure}
 We first show that if $\model$ satisfies the path-dependent counterfactual fairness notion then $\model \in \nullobsunconf$.  The path-dependent counterfactual fairness notion implies that for all $s,d$ such that $P_{\model}(s,d)>0$, 
\begin{equation*}
P_{\model}\Paren{\outcome^{\doop{\sex = \lsex', \dept = \ldept}}=1 \mid \dept = \ldept, \sex=\lsex} = P_{\model}\Paren{\outcome = 1\mid \dept = \ldept, \sex=\lsex}.
\end{equation*}
If for $s' \neq s$, $P_{\model}(s',d)>0$, then we simplify $ P_{\model}\Paren{\outcome^{\doop{\sex = \lsex', \dept = \ldept}}=1 \mid \dept = \ldept, \sex=\lsex}$ using the twin network in Figure~\ref{fig:twin_net_pathwise}.
\begin{align*}
P_{\model}\Paren{\outcome^{\doop{\sex = \lsex', \dept = \ldept}}=1 \mid \dept = \ldept, \sex=\lsex} &= P_{\model^{\text{twin}}}\Paren{\outcome'=1 \mid \doop{\dept' = \ldept, \sex' = \lsex'}, \dept = \ldept, \sex=\lsex} \\
    &\stackrel{(a)}{=} P_{\model^{\text{twin}}}\Paren{\outcome'=1 \mid \doop{\dept' = \ldept, \sex' = \lsex'}}, \\
    &\stackrel{(b)}{=} P_{\model}\Paren{\outcome=1 \mid \dept = \ldept, \sex = \lsex'},
\end{align*}
where $(a)$ follows since $S,D \indep A'$ in the intervened twinned SCM and $(b)$ follows from the twinned SCM. This implies that $$P_{\model}\Paren{\outcome=1 \mid \dept = \ldept, \sex = \lsex} = P_{\model}\Paren{\outcome=1 \mid \dept = \ldept, \sex = \lsex'} = P_{\model}\Paren{\outcome=1 \mid \dept = \ldept}.$$ If instead, $P_{\model}(s',d)=0$, then still $P_{\model}\Paren{A=1 \mid D=d, S=s} = P_{\model}\Paren{A=1 \mid D=d}$. Therefore, $\model \in \nullobsunconf$.

Clearly, if $\model \in \nullgraphunconf$, path-dependent counterfactual fairness is satisfied. However, note that, Example~\ref{ex:unconfexample} satisfies path-dependent counterfactual fairness but does not belong to $\nullgraphunconf$. The conclusion follows from Theorem~\ref{thm:unconf_test_equiv}. 
\end{proof}

\subsection{With Confounding}
In this section, we compare the statistical tests that result from the NDE that the notions of \citet{NabiShpitser18} and \citet{Chiappa19} are based on, and \cite{KusnerLRS17}'s counterfactual fairness and path-dependent counterfactual fairness notions, when confounding is allowed. 

\subsubsection{NDE} With confounding between the mediator and the outcome, \citet{KaufmanKMGP05} obtain bounds on the NDE for the all-binary variable case by the linear programming approach of \citet{BalkePearl97}. The resulting bounds are implied by the IV inequalities but not equivalent to them, which gives us a strictly weaker test than the IV inequalities. For completeness, we present the bounds below. The lower and upper bounds for $NDE(A;0 \rightarrow 1)$ are

\begin{align*}
   &\max \left.\begin{cases}
        P(A=0 \mid S=0) - 1, \\
        P(A=0 \mid D=0 \mid S=0) - P(A=1 \mid D=1 \mid S=0) + P(A=1 \mid D=0 \mid S=1) - 1 \\
        P(A=0 \mid D=1 \mid S=0) - P(A=1 \mid D=0 \mid S=0) + P(A=1 \mid D=1 \mid S=1) - 1
    \end{cases} \right\}, \\
    &\min \left.\begin{cases}
        1- P(A=1 \mid S=0), \\
        1+P(A=0 \mid D=1 \mid S=0) - P(A=1 \mid D=0 \mid S=0) - P(A=0 \mid D=0 \mid S=1) \\
        1+P(A=0 \mid D=0 \mid S=0) - P(A=1 \mid D=1 \mid S=0) - P(A=0 \mid D=1 \mid S=1)
    \end{cases} \right\}. \\
\end{align*}

The lower and upper bounds for $NDE(A;1 \rightarrow 0)$ are

\begin{align*}
   &\max \left.\begin{cases}
        P(A=0 \mid S=1) - 1, \\
        P(A=1 \mid D=0 \mid S=0) - P(A=1 \mid D=1 \mid S=1) + P(A=0 \mid D=0 \mid S=1) - 1 \\
        P(A=1 \mid D=1 \mid S=0) - P(A=1 \mid D=0 \mid S=1) + P(A=0 \mid D=1 \mid S=1) - 1
    \end{cases} \right\}, \\
    &\min \left.\begin{cases}
        1- P(A=1 \mid S=1), \\
        1+P(A=0 \mid D=0 \mid S=1) - P(A=0 \mid D=1 \mid S=0) - P(A=1 \mid D=1 \mid S=1) \\
        1+P(A=0 \mid D=1 \mid S=1) - P(A=0 \mid D=0 \mid S=0) - P(A=1 \mid D=0 \mid S=1) 
    \end{cases} \right\}. \\
\end{align*}
Equating the NDE to $0$, gives us a strictly larger null hypothesis compared to the one obtained based on the IV inequalities. This implies that the resulting statistical test is strictly weaker.

\subsubsection{Counterfactual Fairness \citep{KusnerLRS17}} \label{app:ctrfkusner-cf}
The proof of Proposition~\ref{prop:ctrf-demo-nocf} also holds when confounding is allowed since the implications of the do-calculus rules in the proof hold even in the twin network with confounding. 


\subsubsection{Path-dependent Counterfactual Fairness \citep{KusnerLRS17}} \label{app:pathwise-cf}

 \begin{proposition}\label{prop:cf-ctrf-pathwise}
If $\model \in \modelsedge$ satisfies path-dependent counterfactual fairness then $P_M(D,A,S) \in \distgraph = \distiv$. If $M \in \nullgraph$, then $\model$ satisfies path-dependent counterfactual fairness.
 \end{proposition}
\begin{proof}
We show that a model $\model \in \modelsedge$ that satisfies the path-dependent counterfactual fairness notion is observationally equivalent to a model in $\nullgraph$. This implies that the set of observational distributions of models that satisfy the path-dependent counterfactual notion of fairness, are described by $\distgraph$. 

If $\model \in \modelsedge$ satisfies path-dependent counterfactual fairness, then for all $s,d$ such that $P_{\model}\Paren{s,d}>0$,
\begin{equation*}
P_{\model}\Paren{\outcome^{\doop{\sex = \lsex', \dept = \ldept}}=1 \mid \dept = \ldept, \sex=\lsex} = P_{\model}\Paren{\outcome^{\doop{\sex = \lsex, \dept = \ldept}}=1 \mid \dept = \ldept, \sex=\lsex}
\end{equation*}
for $s' \neq s$. Note that the right-hand side above is $P_{\model}\Paren{\outcome = 1\mid \dept = \ldept, \sex=\lsex}$. The counterfactual $P_{\model}\Paren{\outcome^{\doop{\sex = \lsex', \dept = \ldept}}=1 \mid \dept = \ldept, \sex=\lsex}$ is given by the push-forward of $P(U_A,U|D=d,S=s)$ through $f_{A}(s',d,U_A,U)$. Because of the independence on the value of $s'$, the same holds for the function $$\bar{f}_{A}\Paren{d,U_A,U} = \frac{1}{2}\Paren{f_{A}(0,d,U_A,U) + f_{A}(1,d,U_A,U)}.$$ Consider an SCM, $\bar{\model}$, that is identical to $\model$ except for the causal mechanism of $A$ being $\bar{f}_{A}$. 
Clearly, $\bar{\model} \in \nullgraph$ and $P_M(S,D) = P_{\bar{M}}(S,D)$. By the above argument, for all $s,d$ such that $P_{\model}(s,d)>0$, we have $P_{\model}\Paren{\outcome = 1\mid \dept = \ldept, \sex=\lsex} = P_{\bar{\model}}\Paren{\outcome = 1\mid \dept = \ldept, \sex=\lsex}$. This implies that $P_{\model}(D,A,S) \in \distgraph = \distiv$. 
\begin{figure}[th]
     \centering
            \begin{tikzpicture}
            \tikzstyle{vertex}=[circle,fill=none,draw=black,minimum size=17pt,inner sep=0pt]
\node[vertex] (S) at (0,0) {$S$};
\node[vertex][fill=lightgray] (U_S) at (0,-1) {$U_S$};
\node[vertex] (S') at (0,-2) {$S'$};
\node[vertex] (A) at (3,0) {$A$};
\node[vertex] (A') at (3,-2) {$A'$};
\node[vertex] (D) at (1.5,1.5) {$D$};
\node[vertex] (D') at (1.5,-3.5) {$D'$};
\node[vertex][fill=lightgray] (U_D) at (1.5,-1) {$U_D$};
\node[vertex][fill=lightgray] (U_A) at (3,-1) {$U_A$};
\node[vertex][fill=lightgray] (U) at (5.5,-1) {$U$};

\path (S) edge (D);
\path (D) edge (A);
\path (D') edge (A');
\path (U_S) edge (S);
\path (U_D) edge (D);
\path (U_A) edge (A);
\path (U_A) edge (A');
\path (U) edge (A');
\path (U) edge (A);
\path (U) edge (D);
\path (U) edge (D');
            \end{tikzpicture}
        \caption{Causal graph of twin network $(\model^{\text{twin }})^{\doop{S'=s',D'=d}}$ for $\model \in \nullgraph$} 
        \label{fig:twin_net_pathwise_cf}
        \end{figure}
Conversely, if $\model \in \nullgraph$, then from the twin network of $\model$ in Figure~\ref{fig:twin_net_pathwise_cf}, clearly $A' \indep S' \mid S,D$ and therefore, path-dependent counterfactual fairness is satisfied.   

\end{proof}
Therefore, for the assumed model class, a valid statistical test for path-dependent counterfactual fairness is also a valid test for the IV inequalities from Theorem~\ref{thm:iv_tight} and vice versa.

\section{Additional Results for Bayesian Testing Procedure: Cum-laude Dataset and Prior Sensitivity}\label{app:bayes}

We consider the dataset from \citet{Bol23} that contains data from $5239$ PhD students in the Netherlands studying at a large Dutch university from 2011-2021. \citet{Bol23} observed a bias in the percentage of `cum-laude' distinctions awarded to male PhD students ($6.57 \%$) versus female PhD students ($3.68 \%$). 

As in the Berkeley example, there is data on the sex of the student, their academic field and whether they were awarded cum-laude. Unlike the Berkeley example, there are more covariates that measure additional information, including the sex composition of the dissertation committee, the sex composition of the supervisory team that includes the promoters and co-promoters. For the current analyses we don't take into account these covariates and only analyze the dataset with respect to sex, academic field and award outcome. 

As reported by \citet{Bol23}, unlike the Berkeley dataset, the bias among female and male cum-laude award rates does not vanish when conditioned on department. Therefore, with the assumption of no confounding between the academic field choice and the cum-laude award outcome, the conclusion of the conditional independence test implies that the data generating mechanism is unfair when assuming that no latent unprotected mediators between sex and cum-laude award rates exist. Allowing for confounding requires us to use the Bayesian testing procedure proposed in Section~\ref{sec:bayesiantest} for the IV inequalities. We have $|\cX| = |\cX_{\dept}| = 6$, $|\cY| = |\cX_{\outcome}|=2, |\cZ| = |\cX_{\sex}|=2$. We choose a flat Dirichlet prior over parameters ($\theta = \left \lbrace P(d,a,s): d \in \cX_D, a \in \cX_A, s \in \cX_S\right \rbrace$) in both models $\MM_0, \MM_1$, i.e., for $i=0,1, \pi(\theta | \MM_i) = c_i \text{Dir}\Paren{1,1,\cdots,1}$ where $c_i$ is a normalizing constant. The counts from the data $R_1, R_2 \cdots R_m$ are used to obtain the posterior, $P(\theta \mid R_1, R_2, \cdots R_m)$ which is also a Dirichlet distribution. Using $n=10^6$ samples, we observe no violations of the IV inequality. Therefore, the confidence interval for the posterior probability of the cum-laude data satisfying the IV inequalities is $\left[1-3.69\times 10^{-6},1\right]$. Hence, when allowing for confounding, on arrives at the conclusion that, given the available data and restricting the analysis to only three variables, the fairness of the data-generating mechanism is undecidable.

\textbf{Prior Sensitivity: }
For both the Berkeley dataset and the Bol dataset, the final confidence interval is dependent on the choice of the prior. We presented the analysis with a flat Dirichlet prior, for both datasets. We find that the lower limit of the confidence interval does not change as we vary the parameter $\alpha$ over the interval $\left[ 10^{-2}, 10^{5} \right]$ for a Dirichlet $\text{Dir}\Paren{\alpha, \alpha, \cdots, \alpha}$ prior. 

\textbf{Frequentist Test of \citet{WangRR17}: } The frequentist test of \citet{WangRR17} converts every IV inequality into a one-sided association test for a $2\times2$ contingency table. Specifically, for fixed $d,a$, an IV inequality of the form 
$$\Pr\Paren{D=d,A=a \mid S=1} + \Pr\Paren{D=d,A=1-a \mid S=0} \leq 1$$ is transformed into $$\gamma^{d,a} \leq 0$$ where 
$\gamma^{d,a} \triangleq \Pr\Paren{Q^{d,a}=1\mid S=1} - \Pr\Paren{Q^{d,a}=1\mid S=0}$ where $$Q^{d,a} = \begin{cases}
    \bm{1}\left[D=d,A=a\right] & \text{if }S=1, \\
    1-\bm{1}\left[D=d,A=1-a\right] & \text{if } S=0. \\
\end{cases}$$
Note that $Q^{d,a}$ and $S$ are binary random variables and $\gamma^{d,a}=0$ if and only if $Q^{d,a} \indep S$. Further $\gamma^{d,a} \in [-1,1]$ for all $d,a$. 

Since the direction of the one-sided test matters, we check the sign of the difference of conditional probabilities using maximum likelihood (ML) estimates and then conduct a Pearson's chi-square test for independence. For the Berkeley data, the ML estimates of $\gamma^{d,a}$ were negative (less than $-0.6$) for all $d,a$, and the independence tests rejected the null hypothesis of independence with p-value $0.0$ (i.e., less than the smallest positive number representable using double precision floating point format, i.e., $<5\times10^{-324}$). We take this to be significant evidence that the null hypothesis of $\gamma^{d,a} \leq 0$ is not rejected. As noted in \citet{WangRR17}, although we test for multiple IV inequalities, the Bonferroni correction is $1/2$ and does not scale as the number of IV inequalities.

\end{document}